\documentclass[a4paper,11pt]{paper}

\usepackage[margin=1in]{geometry}
\usepackage{amssymb,amsthm,amsmath}
\usepackage{enumerate}
\usepackage[usenames]{xcolor}
\usepackage[noend,linesnumbered,boxed]{algorithm2e}
\usepackage{xspace}
\usepackage{graphicx}
\usepackage{cite}
\usepackage{paralist}

\usepackage[T1]{fontenc}

\newcommand{\mathset}[1]{\ensuremath {\mathbb {#1}}}
\newcommand{\N}{\mathset {N}}
\newcommand{\eps}{\varepsilon}
\newcommand{\R}{\mathset{R}}

\newcommand{\Z}{\mathset{Z}}

\DeclareMathOperator{\diag}{\text{diag}}

\DeclareMathOperator{\dist}{d}

\DeclareMathOperator{\aff}{aff}
\newcommand{\lowdimapprox}[1][]{(4k+3)(d#1-k-1)+\sqrt{k+1}}

\newtheorem{theorem} {Theorem}[section]

\newtheorem{lemma}[theorem]{Lemma}

\newtheorem{corollary}[theorem]{Corollary}
\newtheorem{claim}[theorem]{Claim}

\newtheorem{remark}[theorem]{Remark}

\newenvironment{alg}{\begin{algorithm}[htbp]}{\end{algorithm}}

\title{Approximate $k$-flat Nearest Neighbor
  Search\thanks{WM and PS
were supported in part by DFG Grants MU 3501/1 and MU 3501/2.
YS was supported by the Deutsche Forschungsgemeinschaft within
the research training group `Methods for Discrete Structures'
(GRK 1408).}}

\author
{
Wolfgang Mulzer\thanks{Institut f\"ur Informatik,
Freie Universit\"at Berlin,
\{\texttt{mulzer,pseiferth,yannikstein}\}\texttt{@inf.fu-berlin.de}.}
\and
Huy L. Nguy\~{\^{e}}n\thanks{Simons Institute, UC Berkeley
\texttt{hlnguyen@cs.princeton.edu}.}
\and
Paul Seiferth\footnotemark[2]
\and
Yannik Stein\footnotemark[2]
}

\begin{document}
\maketitle

\begin{abstract}
Let $k$ be a nonnegative integer.
In the \emph{approximate $k$-flat nearest neighbor}
($k$-ANN) problem, we are given a set
$P \subset \R^d$ of $n$ points in
$d$-dimensional space and a fixed approximation
factor $c > 1$. Our goal is to
preprocess $P$ so that we can efficiently
answer \emph{approximate $k$-flat nearest
neighbor queries}: given a $k$-flat $F$,
find a point in $P$ whose distance to
$F$ is within a factor  $c$ of the
distance between $F$ and the closest
point in $P$. The case $k = 0$ corresponds to
the well-studied approximate nearest neighbor
problem, for which a plethora of results are known,
both in low and high dimensions.
The case $k = 1$ is called \emph{approximate line
nearest neighbor}. In this case, we are aware of only
one provably efficient data structure, due to
Andoni, Indyk, Krauthgamer,
and Nguy\~{\^{e}}n (AIKN)~\cite{AndoniInKrNg09}.
For $k \geq 2$, we know of no previous results.

We present the first efficient data structure that
can handle approximate nearest neighbor queries for
arbitrary $k$. We use a data structure
for $0$-ANN-queries as a black box, and the performance
depends on the parameters of the $0$-ANN
solution: suppose we have an $0$-ANN structure
with query time $O(n^{\rho})$ and space requirement
$O(n^{1+\sigma})$, for $\rho, \sigma > 0$.
Then we can answer $k$-ANN queries in time
$O(n^{k/(k + 1 - \rho) + t})$ and space
$O(n^{1+\sigma k/(k + 1 - \rho)} + n\log^{O(1/t)} n)$.
Here, $t > 0$ is an arbitrary constant and the $O$-notation
hides exponential factors in $k$, $1/t$, and $c$ and
polynomials in $d$.

Our approach generalizes the techniques of AIKN for
$1$-ANN: we partition $P$ into \emph{clusters} of increasing
radius, and we build a low-dimensional data structure for
a random projection of $P$. Given a query flat $F$, the
query can be answered directly in clusters whose radius is
``small'' compared to $d(F, P)$ using a grid. For the
remaining points, the low dimensional approximation turns out to
be precise enough.
Our new data structures also give an improvement in the space
requirement over the previous result for $1$-ANN: we can achieve
near-linear space and sublinear query time, a further step
towards practical applications where space  constitutes
the bottleneck.
\end{abstract}
\newpage
\setcounter{page}{1}

\section{Introduction}

Nearest neighbor search is a fundamental
problem in computational geometry, with
countless applications in databases,
information retrieval, computer vision,
machine learning, signal processing,
etc.~\cite{Indyk04}.  Given a set
$P \subset \R^d$ of $n$ points in
$d$-dimensional
space, we would like to preprocess $P$
so that for any query point $q \in \R^d$,
we can quickly find the point in $P$
that is closest to $q$.

There are efficient
algorithms if the dimension $d$ is
``small''~\cite{Clarkson88,Meiser93}.
However, as $d$ increases, these
algorithms quickly become inefficient:
either the query time approaches linear
or the space grows exponentially with
$d$. This phenomenon is usually called
the ``curse of dimensionality''.
Nonetheless, if one is satisfied with
just an \emph{approximate} nearest
neighbor whose distance to the query
point $q$ lies within some factor
$c=1+\eps$, $\eps > 0$, of the distance between
$q$ and the actual nearest neighbor,
there are efficient solutions even for high
dimensions. Several methods are known,
offering trade-offs between the
approximation factor, the space
requirement, and the query time (see,
e.g.,~\cite{Andoni09,AndoniInNgRa14}
and the references therein).

From a practical perspective, it is
important to keep both the query time
and the space small. Ideally, we would
like algorithms with almost linear
(or at least sub-quadratic) space
requirement and sub-linear query time.
Fortunately, there are solutions with
these guarantees. These methods
include \emph{locality sensitive hashing}
(LSH)~\cite{IndykMo98,KushilevitzOsRa98}
and a more recent approach that improves
upon LSH~\cite{AndoniInNgRa14}. Specifically,
the latter algorithm achieves query time
$n^{7/(8c^{2})+O(1/c^{3})}$ and space
$n^{1+7/(8c^2)+O(1/c^{3})}$, where
$c$ is the approximation factor.

Often, however, the query object is more
complex than a single point. Here,
the complexity of the problem is much less
understood. Perhaps the simplest such
scenario occurs when the query object is
a $k$-dimensional flat, for some small
constant $k$. This is called the
\emph{approximate $k$-flat nearest neighbor}
problem~\cite{AndoniInKrNg09}. It constitutes
a natural generalization of approximate
nearest neighbors, which corresponds
to $k=0$.  In practice, low-dimensional
flats are used to model data subject
to linear variations.  For example, one
could capture the appearance of a
physical object under different lighting
conditions or under different
viewpoints (see~\cite{BasriHaZe07} and the references therein).

So far, the only known algorithm with
worst-case guarantees is for $k=1$, the
\emph{approximate line nearest neighbor}
problem. For this case, Andoni, Indyk,
Krauthgamer, and Nguy\~{\^{e}}n (AIKN) achieve
sub-linear query time $d^{O(1)} n^{1/2+t}$
and space $d^{O(1)} n^{O(1/\eps^2 + 1/t^2)}$,
for arbitrarily small $t > 0$. For the
``dual'' version of the problem, where
the query is a point but the data set
consists of $k$-flats, three results are
known~\cite{BasriHaZe07,Mahabadi15,Magen07}. The
first algorithm is essentially a
heuristic with some control of the
quality of
approximation~\cite{BasriHaZe07}.
The second algorithm provides
provable guarantees and a very fast
query time of
$(d+\log n + 1/\eps)^{O(1)}$~\cite{Magen07}.
The third result, due to Mahabadi, is very recent
and improves the space requirement of Magen's
result~\cite{Mahabadi15}. Unfortunately, these algorithms
all suffer from very high space requirements, thus limiting
their applicability in practice. In fact, even the basic LSH
approach for $k=0$ is already too expensive for large
datasets and additional theoretical work and heuristics are
required to reduce the memory usage and make LSH suitable
for this setting~\cite{Panigrahy06,LvJoWaChLi07}.
For $k \geq 2$, we know of no results in the
theory literature.

\noindent
\textbf{Our results.}
We present the first efficient data structure
for general approximate $k$-flat nearest
neighbor search. Suppose we have a data structure
for approximate \emph{point} nearest neighbor
search with query time $O(n^{\rho} + d\log n)$ and space
$O(n^{1+\sigma} + d\log n)$, for some constants $\rho, \sigma > 0$.
Then our algorithm achieves
query time $O(d^{O(1)}n^{k/(k + 1 - \rho) + t})$
and space
$O(d^{O(1)} n^{1+\sigma k/(k+1-\rho)} + n\log^{O(1/t)} n)$, where
$t > 0$ can be made arbitrarily small.
The constant factors for the query time depend on
$k$, $c$, and $1/t$. Our main result is as follows.

\begin{theorem}\label{thm:main}
  Fix an integer $k \geq 1$ and an approximation factor $c > 1$.
  Suppose we have a data structure
  for approximate \emph{point} nearest neighbor
  search with query time $O(n^{\rho} + d\log n)$ and space
  $O(n^{1+\sigma} + d \log n)$, for some constants $\rho, \sigma > 0$.
  Let $P \subset \R^d$ be a $d$-dimensional $n$-point set.
  For any parameter $t > 0$, we
  can construct a data structure with
  $O(d^{O(1)} n^{1+ k \sigma /(k+1 - \rho)} + n\log^{O(1/t)} n)$ space
  that can answer the following queries in
  expected time $O(d^{O(1)}n^{k/(k+1-\rho) + t})$:
  given a $k$-flat $F \subset \R^d$, find a point $p \in P$
  with $d(p, F) \leq cd(P, F)$.
\end{theorem}

\begin{center}
\begin{tabular}{|c|c|c|}
  \hline
  Algorithm & $\rho$ & $\sigma$\\
  \hline
  AINR~\cite{AndoniInNgRa14} & $7/8c^2 + O(1/c^3)$&
  $7/8c^2 + O(1/c^3)$ \\
  LSH1~\cite[Theorem~3.2.1]{Andoni09} & $1/c^2$ & $1/c^2$\\
  LSH2~\cite[Theorem~3.4.1]{Andoni09} & $O(1/c^2)$ & $0$\\
  \hline
\end{tabular}
\end{center}
The table above gives an overview of
some approximate point nearest neighbor structures
that can be used in Theorem~\ref{thm:main}.
The result by AINR gives the current best query performance for large
enough values of $c$. For smaller $c$, an approach
using locality sensitive hashing (LSH1) may be preferable.
With another variant of locality sensitive hashing (LSH2),
the space can be made almost linear, at the expense of a slightly
higher query time. The last result (and related practical
results, e.g.,~\cite{LvJoWaChLi07}) is of particular
interest in applications as the memory consumption is a major bottleneck in practice. It also improves
over the previous algorithm by AIKN for line queries.

Along the way towards Theorem~\ref{thm:main},
we present a novel data structure for
$k$-flat \emph{near} neighbor reporting when the dimension $d$ is
constant. The space requirement in this case is
$O_d(n \log^{O(d)} n)$ and the query time is
$O_d(n^{k/(k+1)}\log^{d-k-1} n + |R|)$, where $R$ is the answer set.
We believe that this data structure may be of independent
interest and may lead to further applications.
Our results provide a vast generalization of the result in  AIKN
and shows for the first time
that it is possible to achieve provably efficient
nearest neighbor search for higher-dimensional query objects.

\noindent
\textbf{Our techniques.}
Our general strategy is similar to the approach by AIKN.
The data structure  consists of two  main structures: the
\emph{projection structure} and the \emph{clusters}.
The projection structure works by projecting the point set
to a space of constant dimension and by answering the nearest
neighbor query in that space. As we will see, this suffices
to obtain a rough estimate for the distance, and it can be
used to obtain an exact answer if the point set is ``spread
out''.

Unfortunately, this does not need to be the case.
Therefore, we partition the point set into a sequence
of \emph{clusters}.
A cluster consists of $m$ points and a $k$-flat $K$
such that all points in the cluster are ``close'' to $K$,
where $m$ is a parameter to be optimized.
Using a rough estimate from the projection structure, we can
classify the clusters as \emph{small} and \emph{large}.
The points in the large clusters are spread out and
can be handled through projection. The points in the small
clusters are well behaved and can be handled directly
in high dimensions using grids and discretization.

\noindent
\textbf{Organization.}
In order to provide the curious reader with quick gratification,
we will give the main data structure together with the properties
of the cluster and the projection structure in
Section~\ref{sec:mainDS}.  Considering these structures as
black boxes, this already proves Theorem~\ref{thm:main}.

In the remainder of the paper, we describe the details of the
helper structures. The necessary tools are introduced
in Section~\ref{sec:preliminaries}. Section~\ref{sec:clusterstructure}
gives the approximate nearest neighbor algorithm for small clusters.
In Section~\ref{sec:lowdimstructure}, we consider approximate
near neighbor reporting for $k$-flats in constant dimension.
This data structure is then used for the projection
structures in Section~\ref{sec:projectionstructures}.

\section{Main Data Structure and Algorithm Overview}
\label{sec:mainDS}

We describe our main data structure for approximate
$k$-flat nearest neighbor search. It relies on various substructures
that will be described in the following sections.
Throughout, $P$ denotes a $d$-dimensional $n$-point set,
and $c > 1$ is the desired approximation factor.

Let $K$ be a $k$-flat in $d$ dimensions.
The \textit{flat-cluster} $C$ (or cluster for short)
of $K$ with radius $\alpha$ is
the set of all points with distance at most $\alpha$ to $K$, i.e.,
$C = \{ p \in \R^d \mid d(p,K) \leq \alpha \}$.
A cluster is \emph{full}
if it contains at least $m$ points from $P$, where $m$ is a
parameter to be determined. We call $P$
\textit{$\alpha$-cluster-free} if there is no full cluster with
radius $\alpha$.
Let $t > 0$ be an arbitrarily small parameter.
Our data structure requires the following three subqueries.
\begin{enumerate}[\bfseries Q1:]
 \item
\label{itm:q1}
Given a query flat $F$,
find a point  $p \in P$ with $d(p,F) \leq n^t d(P, F)$.
\item
\label{itm:q2}
Assume $P$ is contained in a flat-cluster with radius $\alpha$.
Given a query flat $F$ with  $d(P,F) \geq \alpha/n^{2t}$,
return a point $p \in P$ with $d(p,F) \leq
cd(P,F)$.
\item
\label{itm:q3}
Assume $P$ is $\alpha n^{2t}/(2k+1)$-cluster free.
Given  a query flat $F$ with
$d(P,F) \leq \alpha$, find the nearest neighbor $p^*\in P$ to $F$.
\end{enumerate}

Briefly, our strategy is as follows: during the preprocessing
phase, we partition the point set into a set of full
clusters of increasing radii. To answer a query $F$, we first perform
a query of type \textbf{Q1} to obtain an $n^t$-approximate
estimate $\widetilde{r}$ for $d(P, F)$. Using $\widetilde{r}$,
we identify the ``small'' clusters. These clusters can be
processed using a query of type \textbf{Q2}.
The remaining point set contains no ``small'' full cluster, so we
can process it with a query of type \textbf{Q3}.

We will now describe the properties of the subqueries and the
organization of the data structure in more detail.
The data structure for \textbf{Q2}-queries is called the
\emph{cluster structure}. It is described in
Section~\ref{sec:clusterstructure}, and it has
the following properties.

\begin{theorem}\label{thm:q2}
  Let $Q$ be a $d$-dimensional $m$-point set
  that is contained in a flat-cluster of radius $\alpha$.
  Let $c > 1$ be an approximation factor.
  Using space
  $O_c(m^{1+\sigma} + d\log^2 m)$, we can build a data
  structure with the following property.
  Given a query $k$-flat $F$ with $d(P, F) \geq \alpha/n^{2t}$
  and an estimate $\widetilde{r}$ with
  $d(P, F) \in [\widetilde{r}/n^t, \widetilde{r}]$,
  we can find
  a $c$-approximate nearest neighbor for $F$ in $Q$ in total
  time
  $O_c((n^{2t}k^2)^{k+1}(m^{1-1/k+\rho/k} + (d/k)\log m))$.
\end{theorem}
The data structures for \textbf{Q1} and \textbf{Q3} are very similar,
and we cover them in Section~\ref{sec:projectionstructures}.
They are called \emph{projection structures}, since
they are based on projecting $P$ into a low dimensional
subspace.
In the projected space, we use a
data structure for approximate $k$-flat \textit{near}
neighbor search to be described in
Section~\ref{sec:lowdimstructure}.
The projection structures have the following properties.
\begin{theorem}
 \label{thm:q1}
 Let $P$ be a $d$-dimensional $n$-point set, and let $t > 0$ be
 a small enough constant.
Using space and time $O(n \log^{O(1/t)}n)$, we can obtain
a data structure for the following query:
given a $k$-flat $F$, find a point $p \in P$ with
$d(p, F) \leq n^td(P, F)$.
A query needs $O(n^{k/(k+1)}\log^{O(1/t)} n)$ time, and
the answer is correct with high probability.
\end{theorem}

\begin{theorem}
 \label{thm:q3}
 Let $P$ be a $d$-dimensional $n$-point set, and let $t > 0$ be
 a small enough constant.
Using space and time $O(n \log^{O(1/t)}n)$, we can obtain
a data structure for the following query:
given a $k$-flat $F$ and $\alpha > 0$ such that
$d(F, P) \leq \alpha$ and such that $P$ is
$\alpha n^t/(2k+1)$-cluster-free, find an exact nearest neighbor for
$F$ in $P$. A query needs $O(n^{k/(k+1)}\log^{O(1/t)}n + m)$
time, and the answer is correct with high probability.
Here, $m$ denotes the size of a full cluster.
\end{theorem}

\subsection{Constructing the Data Structure}

First, we build a projection structure for \textbf{Q1} queries
on $P$. This needs $O(n\log^{O(1/t)}n)$ space,
by Theorem~\ref{thm:q1}.
Then, we repeatedly find the full flat-cluster $C$ with smallest
radius.
The $m$ points in $C$ are removed from $P$, and we build
a cluster structure for \textbf{Q2} queries on this set.
By Theorem~\ref{thm:q2}, this needs
$O_c(m^{1 + \sigma} + d\log^2 m)$ space.
To find $C$, we check all flats $K$ spanned
by $k+1$ distinct points of $P$. In Lemma~\ref{lem:pert} below,
we prove that this provides a good enough approximation.
In the end, we have $n/m$ point sets
$Q_1,\dots,Q_{n/m}$ ordered by decreasing radius,
i.e., the cluster for $Q_1$ has the largest radius.
The total space occupied  by the cluster structures is
$O(nm^{\sigma} + (n/m)d\log^2 n)$.

Finally, we build a perfect binary tree $T$ with $n/m$ leaves
labeled $Q_1,\dots,Q_{n/m}$, from left to right.
For a node $v \in T$
let $Q_v$ be the union of all $Q_i$ assigned to leaves below $v$.
For each $v \in T$ we build
a data structure for $Q_v$ to answer \textbf{Q3} queries.
Since each point is contained in $O(\log n)$ data structures,
the total size is $O(n\log^{O(1/t)} n)$, by
Theorem~\ref{thm:q3}.
For pseudocode, see Algorithm~\ref{alg:high-dim-preprocessing}.

\begin{alg}
\DontPrintSemicolon
\SetKw{DownTo}{downto}
\SetKwInOut{Input}{Input}
\Input{point set $P\subset \R^d$, approximation factor $c$,
parameter $t > 0$}
$Q \leftarrow P$\;
\For{$i\gets n/m$ \DownTo $1$}{
For each $V \in \binom{Q}{k+1}$, consider
the $k$-flat $K_V$ defined by $V$. Let $\alpha_{V}$ be the
radius of the smallest flat-cluster of $K_V$ with
exactly $m$ points of $Q$. \;
Choose the flat $K = K_V$ that minimizes
$\alpha_V$ and set $\alpha_i = \alpha_V$. \;
Remove from $Q$ the set $Q_i$ of $m$ points in
$Q$ within distance $\alpha_i$ from $K$.\;
Construct a cluster structure $C_i$ for the cluster
$(K,Q_i)$.\;
}
Build a perfect binary tree $T$ with $n/m$ leaves,
labeled $Q_1,\dots,Q_{n/m}$ from left to right.\;
\ForEach{node $v \in T$}{
  Build data structure for \textbf{Q\ref{itm:q3}}  queries as in
  Theorem~\ref{thm:q3}
  for the set $Q_v$ corresponding to the leaves below $v$.
}
\caption{Preprocessing algorithm. Compared
  with AIKN~\cite{AndoniInKrNg09}, we
  organize the  projection structure in a tree to save space.}
\label{alg:high-dim-preprocessing}
\end{alg}

\subsection{Performing a Query}
Suppose we are given a $k$-flat $F$.
To find an approximate nearest neighbor
for $F$ we proceed similarly as AIKN~\cite{AndoniInKrNg09}. We use
\textbf{Q2} queries on
``small'' clusters  and \textbf{Q3} queries on the remaining points;
for pseudocode, see Algorithm~\ref{alg:high-dim-query}.

\begin{alg}
  \DontPrintSemicolon
\SetKwInOut{Input}{Input}
\SetKwInOut{Output}{Output}
  \Input{query flat $F$}
  \Output{a $c$-approximate nearest neighbor for $F$ in $P$}
  Query the root of $T$ for a
  $n^t$-approximate nearest neighbor $p_1$ to $F$.
  \tcc*{type Q\ref{itm:q1}}
  $\widetilde{r} \gets d(p_1, F)$\;
  $i^*\gets$ maximum
  $i\in \{1, \dots, n/m\}$ with $\alpha_{i}> \widetilde{r} n^t$,
  or $0$
  if no such $i$ exists\;

  \For{$i \gets i^*+1$ \KwTo $n/m$}{
    \tcc{type Q\ref{itm:q2}; we have
      $d(Q_i, F) \geq \widetilde{r}/n^t \geq \alpha_i/n^{2t}$}
    Query cluster structure $C_i$ with estimate $\widetilde{r}$.
  }
 \tcc{type Q\ref{itm:q3}}
  Query projection structure for a
  $\widetilde{r}$-thresholded nearest neighbor of $F$ in
  $Q = \bigcup_{i = 1}^{j^*} U_i$.
  \Return closest point to $F$ among query results.
 \caption{Algorithm for finding approximate nearest neighbor in
   high dimensions.}
 \label{alg:high-dim-query}
\end{alg}

First, we perform a query of type \textbf{Q1}
to get a $n^t$-approximate nearest neighbor $p_1$ for $F$ in time
$O(n^{k/(k+1)}\log^{O(1/t)}n)$. Let $\widetilde{r} = d(p_1, F)$.
We use $\widetilde{r}$ as an estimate to distinguish between
``small'' and
``large'' clusters. Let $i^* \in \{1,\dots,n/m\}$
be the largest integer such that the cluster assigned with $Q_{i^*}$
has radius $\alpha_{i^*} > \widetilde{r} n^t$.
For $i = i^* + 1, \dots,  n/m$,
we use $\widetilde{r}$ as an estimate for a \textbf{Q2} query
on $Q_i$. Since
$|Q_i| = m$ and by Theorem~\ref{thm:q2},
this needs
total time
$O(n^{2t(k + 1) + 1}m^{-1/k + \rho/k} + (n/m)d\log^2 m)$.

It remains to deal with points in ``large'' clusters.
The goal is to perform a type \textbf{Q3}
query on $\bigcup_{1 \leq i \leq i^*} Q_i$. For this, we start at the
leaf of $T$ labeled $Q_{i^*}$ and walk up to the root.
Each time we encounter a new node $v$ from its right child,
we perform a \textbf{Q3} query on $Q_u$, where $u$ denotes the
left child of $v$.
Let $L$ be all the left children we find in this way. Then clearly
we have $|L| = O(\log n)$  and
$\bigcup_{u \in L} Q_u = \bigcup_{1 \leq i \leq i^*} Q_i$.
Moreover, by construction, there is no full cluster with radius
less than
$\widetilde{r} n^t$ defined by $k+1$ vertices of $Q_u$
for any $u \in L$.
We will see that this implies every $Q_u$ to be
$\widetilde{r} n^t/(2k+1)$-cluster-free, so
Theorem~\ref{thm:q3} guarantees a total query
time of $O(n^{k/(k+1)}\log^{O(1/t)} n + m)$ for this step.
Among all the points we obtained during the queries,
we return the one that is closest to $F$.
A good trade-off point is achieved for $m = nm^{-1/k+\rho/k}$,
i.e., for $m = n^{k/(k+1-\rho)}$.
This gives the bounds claimed in Theorem~\ref{thm:main}.

\noindent
\textbf{Correctness.}
Let $p^*$ be a point with $d(p^*, F) = d(P,F)$.
First, suppose that $p^* \in Q_i$,
for some $i > i^*$. Then, we have
$d(p^*,F) \geq \widetilde{r}/n^t \geq \alpha_i/n^{2t}$, where $\alpha_i$
is the radius of the cluster assigned to $Q_i$.
Since $\widetilde{r}$ is a valid $n^t$-approximate
estimate for $d(F, Q_i)$,
a query of type \textbf{Q2} on $Q_i$ gives a $c$-approximate nearest
neighbor, by Theorem~\ref{thm:q2}.
Now, suppose that $p^* \in Q_i$ for $1 \leq i \leq i^*$.
Let $u$ be the node of $L$
with $p^* \in Q_u$. Then Theorem~\ref{thm:q3} guarantees
that we will find
$p^*$ when doing a \textbf{Q3} query on $Q_u$.

\section{Preliminaries}
\label{sec:preliminaries}

\noindent
\textbf{Partition Trees.}
Fix an integer constant $r > 0$, and
let $P \subset \R^d$ be a $d$-dimensional $n$-point set.
A \emph{simplicial $r$-partition} $\Xi$ for $P$ is a sequence
$\Xi = (P_1, \Delta_1), \dots, (P_m, \Delta_m)$ of pairs
such that (i) the sets $P_1, \dots, P_m$ form a partition of
$P$ with $n/r \leq |P_i| \leq \lceil 2n/r \rceil$, for
$i = 1, \dots, m$;
(ii) each $\Delta_i$ is a relatively open simplex with
$P_i \subset \Delta_i$, for $i = 1, \dots, m$; and
(iii) every hyperplane $h$ in $\R^d$ \emph{crosses} $O(r^{1-1/d})$
  simplices $\Delta_i$ in $\Xi$. Here, a hyperplane $h$ \emph{crosses}
  a simplex $\Delta$ if $h$ intersects $\Delta$, but does not
  contain it.
In a classic result, Matou\v{s}ek showed that such a simplicial
partition always exists and that it can be computed
efficiently~\cite{Matousek92,Chazelle00}.

\begin{theorem}[Partition theorem, Theorem~3.1 and Lemma~3.4
  in \cite{Matousek92}]\label{thm:partition}
  For any $d$-dimensional $n$-point set $P \subset \R^d$ and for
  any constant $1 \leq r \leq n/2$, there exists a simplicial
  $r$-partition for $P$. Furthermore, if $r$ is bounded by
  a constant, such a partition can be found in time $O(n)$.\qed
\end{theorem}

Through repeated application of Theorem~\ref{thm:partition},
one can construct a \emph{partition tree} for $P$.
A partition tree $\mathcal{T}$ is a rooted tree in which each node
is associated with a pair $(Q, \Delta)$, such that $Q$ is a subset
of $P$ and $\Delta$ is a relatively open simplex that
contains $Q$. If $|Q| \geq 2r$, the children of $(Q, \Delta)$
constitute a simplicial $r$-partition of $Q$. Otherwise,
the node $(Q, \Delta)$ has $|Q|$ children where each child
corresponds to a point in $Q$. A partition tree has
constant degree, linear size, and logarithmic depth.

Given a hyperplane $h$, there is a straightforward query
algorithm to find the highest nodes in $\mathcal{T}$ whose
associated simplex does not cross $h$: start at the root and
recurse on all children whose associated simplex crosses $h$;
repeat until there are no more crossings or until a leaf is reached.
The children of the traversed nodes whose simplices do not cross
$h$ constitute the desired answer. A direct application of
Theorem~\ref{thm:partition} yields a partition tree for which this
query takes time $O(n^{1-1/d+\gamma})$, where $\gamma > 0$ is
a constant that can be made arbitrarily small by increasing $r$.
In 2012, Chan~\cite{Chan12} described a more global
construction that eliminates the $n^\gamma$ factor.

\begin{theorem}[Optimal Partition
   Trees~\cite{Chan12}]\label{thm:optpartition}
  For any $d$-dimensional $n$-point set $P \subset \R^d$,
  and for any large enough constant $r$,
  there is a \emph{partition tree} $\mathcal{T}$ with the following
  properties:
   (i) the tree $\mathcal{T}$ has degree $O(r)$ and depth $\log_r n$;
    (ii) each node is of the form $(Q, \Delta)$,
     where $Q$ is a subset of $P$ and $\Delta$ a relatively open
     simplex that contains $Q$; (iii) for each node $(Q, \Delta)$,
     the simplices of the children of $Q$ are contained in $\Delta$
     and are pairwise disjoint; (iv) the point set associated
     with a node of depth $\ell$ has size at most $n/r^\ell$;
     (v) for any hyperplane $h$ in $\R^d$, the number
     $m_\ell$ of simplices in $\mathcal{T}$  that
     $h$ intersects at level $\ell$ obeys the recurrence
     \[
       m_\ell = O(r^{\ell(d-1)/d} + r^{\ell(d-2)/(d-1)}m_{\ell - 1}
         + r\ell \log r \log n).
     \]
   Thus, $h$ intersects $O(n^{1-1/d})$ simplices in total.
 The tree $\mathcal{T}$ can be build in expected time $O(n\log n)$.
\end{theorem}

\noindent
\textbf{$k$-flat Discretization.}
For our cluster structure we must find $k$-flats that are
close to many points. The following lemma shows that it
suffices to check ``few'' $k$-flats for this.
\begin{lemma}\label{lem:pert}
  Let $P \subset \R^d$ be a finite point set with $|P| \geq k+1$,
  and let $F \subset \R^d$ be a $k$-flat. There is a $k$-flat $F'$
  such that $F'$ is the affine hull of $k + 1$ points in $P$ and
  $\delta_{F'}(P) \leq (2k + 1) \delta_{F}(P)$,
  where $\delta_{F'}(P) = \max_{p \in P}\dist(p, F')$ and
  $\delta_{F}(P) = \max_{p\in P}\dist(p, F)$.
\end{lemma}
\begin{proof}
  This proof generalizes the proof of Lemma~2.3 by
  AIKN~\cite{AndoniInKrNg09}.

  Let $Q$ be the orthogonal projection of $P$ onto $F$.
  We may assume that $F$ is the affine hull of $Q$, since otherwise
  we could replace $F$ by $\text{aff}(Q)$ without affecting
  $\delta_F(P)$. We choose an orthonormal basis for $\R^d$ such that
  $F$ is the linear subspace spanned by the first $k$ coordinates.
  An affine basis for $F'$ is constructed as follows: first,
  take a point $p_0 \in P$ whose $x_1$-coordinate is minimum. Let
  $q_0$ be the projection of $p_0$ onto $F$, and translate the
  coordinate system such that $q_0$ is the origin.
  Next, choose $k$ additional points $p_1, \dots ,p_{k} \in P$ such
  that $|\det(q_1, \dots, q_k)|$ is maximum, where $q_i$ is
  the projection of $p_i$ onto $F$, for $i = 1, \dots, k$. That is,
  we choose $k$ additional points such that the volume of the
  $k$-dimensional parallelogram spanned by their projections
  onto $F$ is maximized. The set $\{q_1, \dots, q_k\}$ is a basis for
  $F$, since the maximum determinant cannot be $0$ by our assumption
  that $F$ is spanned by $Q$.

  Now fix some point $p \in P$ and let $q$ be its projection onto
  $F$. We write $q = \sum_{i=1}^{k} \mu_i q_i$. Then, the point
  $r = \sum_{i=1}^{k} \mu_i p_i + (1-\sum_{i=1}^{k} \mu_i) p_0$
  lies in $F'$. By the triangle inequality, we have
  \begin{align}\label{equ:flat_triang}
    \dist(p, r) \leq \dist(p, q) + \dist(q, r) \leq \delta_F(P) +
      \dist(q, r).
  \end{align}
  To upper-bound $\dist(q, r)$ we first show that all coefficients
  $\mu_i$ lie in $[-1, 1]$.
  \begin{claim}\label{clm:coeffbound}
    Take $p\in P$, $q \in Q$  and $q_1, \dots, q_{k}$
    as above. Write $q = \sum_{i=1}^{k} \mu_i q_i$. Then for
    $i = 1, \dots, k$, we have $\mu_i \in [-1, 1]$, and
    $\mu_j \geq 0$ for at least one  $j\in\{1, \dots, k\}$.
  \end{claim}
  \begin{proof}
    We first prove that all coefficients $\mu_i$ lie in the
    interval $[-1,1]$. Suppose that $|\mu_i| > 1$ for some
    $i \in \{1, \dots, k\}$. We may assume that $i = 1$.
    Using the linearity of the determinant,
    \begin{multline*}
      |\det(q, q_2, \dots, q_k)| = |\det(\mu_1 q_1, q_2, \dots, q_k)|
      = |\mu_1| \cdot |\det(q_1, q_2, \dots, q_k)|
      > |\det(q_1, q_2, \dots, q_k)|,
    \end{multline*}
     contradicting the choice of $q_1, \dots, q_k$.

     Furthermore, by our choice of the origin,
     all points in $Q$ have a non-negative
     $x_1$-coordinate.
     Thus, at least one coefficient $\mu_j$, $j\in\{1, \dots, k\}$,
     has to be non-negative.
  \end{proof}
  Using Claim~\ref{clm:coeffbound}, we can now bound $\dist(q, r)$.
  For $i = 1, \dots, k$, we write
  $p_i = q_i + q_i^{\perp}$, where $q_i^{\perp}$ is orthogonal to $F$.
  Then,
  \begin{multline}\label{equ:perp_triang}
    \dist(q, r) =
    \left\|\sum_{i = 1}^k \mu_i q_i -
        \sum_{i = 1}^{k} \mu_i (q_i + q_i^\perp) -
        \left(1 - \sum_{i =1 }^{k} \mu_i\right)
        p_0\right\| \\
      =\left\| \sum_{i=1}^{k} \mu_i q_i^{\perp} +
      \left(1-\sum_{i=1}^{k} \mu_i\right) p_0 \right\|
    \leq \left(\sum_{i=1}^{k} |\mu_i| +
      \left|1 - \sum_{i=1}^{k} \mu_i\right|
    \right) \delta_F(P) \leq 2k \delta_F(P),
  \end{multline}
  since
  $\|q_1^\perp\|, \dots, \|q_k^\perp\|, \|p_0\| \leq \delta_F(P)$,
  and since $\left|1-\sum_{i=1}^{k} \mu_i\right| \leq k$ follows from
  fact that at least one $\mu_i$ is non-negative.
  By (\ref{equ:flat_triang}) and (\ref{equ:perp_triang}), we get
  $\dist(p, F') \leq (2k+1) \delta_F(P)$.
\end{proof}

\begin{remark}
  For $k=1$, the proof of Lemma~\ref{lem:pert} coincides with
  the proof of Lemma~2.3 by AIKN~\cite{AndoniInKrNg09}. In this
  case, one can obtain a
  better bound on $\dist(q, r)$ since $q$ is a convex combination
  of $q_0$ and $q_1$. This gives $\delta_{F'}(P)\leq 2 \delta_F(P)$.
\end{remark}

\section{Cluster Structure}
\label{sec:clusterstructure}
A \emph{$k$-flat cluster structure} consists of a $k$-flat
$K$ and a set $Q$ of $m$ points with
$d(q, K) \leq \alpha$, for all $q \in Q$.
Let $K: u \mapsto A'u + a$ be a parametrization of $K$, with
$A' \in \R^{d \times k}$ and $a \in \R^d$ such that the columns
of $A'$ constitute an orthonormal basis for $K$ and such that
$a$ is orthogonal to $K$.
We are also given an approximation parameter $c > 1$.
The cluster structure uses a data structure
for approximate point nearest neighbor search as a black box.
We assume that we have such a structure available that
can answer $c$-approximate point nearest neighbor queries
in $d$ dimensions with query time $O_c(n^{\rho} + d\log n)$
and space requirement $O_c(n^{1+\sigma} + d\log n)$ for
some constants $\rho, \sigma > 0$. As mentioned in the
introduction, the literature offers several data structures
for us to choose from.

The cluster structure distinguishes two cases:
if the query flat $F$ is close to $K$, we can approximate
$F$ by few ``patches'' that are parallel to $K$, such that
a good nearest neighbor for the patches is also
good for $K$. Since the patches are parallel to $K$, they
can be handled through $0$-ANN queries in the orthogonal
space $K^\perp$ and low-dimensional
queries inside $K$.
If the query flat is far from $K$, we can approximate $Q$
by its projection onto $K$ and handle the query with a
low-dimensional data structure.

\subsection{Preprocessing}

Let $K^\perp$ be the linear subspace of $\R^d$ that is orthogonal
to $K$. Let $Q_a$ be the projection of $Q$ onto $K$,
and let $Q_b$ be the projection of $Q$ onto $K^\perp$.
We compute a $k$-dimensional partition
tree $\mathcal{T}$ for $Q_a$.
As stated in Theorem~\ref{thm:optpartition}, the tree $\mathcal{T}$
has $O(m)$ nodes, and it can be computed in time $O(m \log m)$.

For each node $(S_a, \Delta)$ of $\mathcal{T}$, we do the following:
we determine the set $S \subseteq Q$ whose projection onto $K$
gives $S_a$, and we take the projection $S_b$ of $S$ onto
$K^\perp$. Then, we build a $d-k$ dimensional $c'$-ANN data
structure for $S_b$, as given by the assumption,
where $c'=(1-1/\log n)c$. See
Algorithm~\ref{alg:preprocessCluster} for pseudocode.

\begin{alg}
\SetKwInOut{Input}{Input}
\SetKwFunction{CreatePartitionStructure}{CreatePartitionStructure}
\SetKwFunction{CreateSlabStructure}{CreateSlabStructure}
\Input{$k$-flat $K \subseteq \R^d$, point set $Q \subset \R^d$ with
  $d(q, K) \leq \alpha$ for all $q \in Q$, approximation
  parameter $c > 1$}

  $Q_a \gets$ projection of $Q$ onto $K$

  $Q_b \gets$ projection of $Q$ onto $K^\perp$

  Build a $k$-dimensional partition tree $\mathcal{T}$ for $Q_a$
   as in Theorem~\ref{thm:optpartition}.

  $c' \gets  (1-1/\log n)c$

  \ForEach{node $(S_a, \Delta) \in \mathcal{T}$} {
      $S_b \gets$ projection of the points in $Q$
      corresponding to $S_a$ onto $K^\perp$

      Build a $(d-k)$-dimensional $c'$-ANN structure for $S_b$
      as given by the assumption.
}
 \caption{CreateClusterStructure}
 \label{alg:preprocessCluster}
\end{alg}

\begin{lemma}\label{lem:clusterPrep}
  The cluster structure can be constructed in total
  time  $O_c(m^{2+\rho} + md\log^2 m)$, and it requires
  $O_c(m^{1+\sigma} + d\log^2 m)$  space.
\end{lemma}
\begin{proof}
By Theorem~\ref{thm:optpartition}, the partition tree can be
built in $O(m \log m)$ time. Thus,
the preprocessing time is dominated by the time to
construct the $c'$-ANN data structures at the nodes
of the partition tree $\mathcal{T}$.
Since the sets on each level of $\mathcal{T}$ constitute a
partition of $Q$, and since the sizes of the sets decrease
geometrically, the bounds on the preprocessing time and
space requirement follow directly from our assumption.
Note that by our choice of $c' = (1-1/\log n)c$, the
space requirement and query time for the
ANN data structure change only by a constant factor.
\end{proof}

\subsection{Processing a Query}

We set $\eps = 1/100 \log n$.
Let $F$ be the query $k$-flat, given as $F: v \mapsto B'v + b$,
with $B' \in \R^{d \times k}$ and $b \in \R^d$ such
that the columns of $B'$ are an orthonormal basis for $F$
and $b$ is orthogonal to $F$. Our first task is to find bases
for the flats $K$ and $F$ that provide us
with information about the relative position of $K$ and $F$.
For this, we take the matrix
$M = A'^T B' \in \R^{k \times k}$, and we compute a \emph{singular
value decomposition} $M = U \Sigma V^T$ of $M$~\cite[Chapter~7.3]{HornJo13}.
Recall that $U$ and $V$ are orthogonal $k \times k$ matrices and
that $\Sigma = \diag(\sigma_1, \dots, \sigma_k)$
is a $k \times k$ diagonal matrix
with $\sigma_1 \geq \dots \geq \sigma_{k} \geq 0$.
We call $\sigma_1, \dots, \sigma_k$ the \emph{singular values}
of $M$. The following lemma summarizes the properties of
the SVD that are relevant to us.
\begin{lemma}\label{lem:svd}
  Let $M = A'^TB'$, and let $M = U \Sigma V^T$ be a singular
  value decomposition for $M$. Let $u_1, \dots, u_k$ be
  the columns of $U$ and $v_1, \dots, v_k$ be the columns of $V$.
  Then,
  (i) $u_1, \dots, u_k$ is an orthonormal basis
  for $K$ (in the coordinate system induced by $A'$);
  (ii) $v_1, \dots, v_k$ is an orthonormal basis for $F$
  (in the coordinate system induced by $B'$): and
  (iii) for $i = 1, \dots, k$, the projection of $v_i$ onto
  $K$ is $\sigma_iu_i$ and the projection of $u_i$ onto
  $F$ is $\sigma_iv_i$ (again in the coordinate systems
  induced by $A'$ and $B'$).
  In particular, we have $\sigma_1 \leq 1$.
\end{lemma}
\begin{proof}
  Properties (i) and (ii) follow since $U$ and $V$
  are orthogonal matrices.
  Property (iii) holds because $M = A'^TB'$ describes the projection
  from $F$ onto $K$ (in the coordinate systems induced by $A'$ and
  $B'$) and because $M^T = B'^TA' = V\Sigma U^T$ describes the
  projection from $K$ onto $F$.
\end{proof}

We reparametrize $K$ according to $U$ and $F$ according to $V$.
More precisely, we set $A = A'U$ and $B = B'V$, and we write
$K: u \mapsto Au + a$ and $F: v \mapsto Bv + b$.
The new coordinate system provides
a simple representation for the distances between $F$ and $K$.
We begin with a technical lemma that is a simple corollary of
Lemma~\ref{lem:svd}.
\begin{lemma}\label{lem:matrices}
  Let $a_1, \dots, a_k$ be the columns of the matrix $A$;
  let $a_1^\parallel, \dots, a_k^{\parallel}$ be the columns of
  the matrix $BB^TA$, and $a_1^\perp, \dots, a_k^\perp$ the
  columns of the matrix $A - BB^TA$.
  Then, (i) for $i = 1, \dots, k$, the vector $a_i^\parallel$ is
  the projection of $a_i$ onto $F$ and the vector $a_i^\perp$
  is the projection of $a_i$ onto $F^\perp$;
  (ii) for $i = 1, \dots, k$, we have
  $\|a_i^\parallel\| = \sigma_i$ and
   $\|a_i^\perp\| = \sqrt{1-\sigma_i}$; and
  (iii) the vectors $a_1^\parallel, \dots, a_k^\parallel,
  a_1^\perp, \dots, a_k^\perp$ are pairwise orthogonal.
  An analogous statement holds for the matrices
  $B$, $AA^TB$, and $B - AA^TB$.
\end{lemma}
\begin{proof}
  Properties (i) and (ii) are an immediate consequence of the
  definition of $A$ and $B$ and of Lemma~\ref{lem:svd}.
  The set $a_1^\parallel, \dots, a_k^\parallel$ is orthogonal
  by Lemma~\ref{lem:svd}(ii). Furthermore, since for any
  $i,j \in \{1, \dots, k\}$, the vector
  $a_i ^\parallel$ lies in $F$ and the vector
  $a_i^\perp$ lies in $F^\perp$,
  $a_i^\parallel$ and $a_j^\perp$ are orthogonal.
  Finally, let $1 \leq i < j \leq k$. Then,
  \[
    \langle a_i^\perp, a_j^\perp\rangle
    =
    \langle a_i^\perp, a_j^\perp\rangle
    + \langle a_i^\perp, a_j^\parallel\rangle
    + \langle a_i^\parallel, a_j^\perp\rangle
    + \langle a_i^\parallel, a_j^\parallel\rangle
    = \langle a_i, a_j \rangle = 0,
  \]
  since we already saw that
  $\langle a_i^\perp, a_j^\parallel \rangle =
  \langle a_i^\parallel, a_j^\perp \rangle =
  \langle a_i^\parallel, a_j^\parallel \rangle =
  \langle a_i, a_j \rangle = 0$.
  The argument for the other matrices is
  completely analogous.
\end{proof}

The next lemma shows how our choice of bases gives a
convenient representation of the distances between $F$ and
$K$.
\begin{lemma}\label{lem:svddistance}
Take two points $x_F \in K$ and $y_K \in F$ such that
$d(F, K) = d(y_K, x_F)$. Write $x_F = Au_F + a$ and
$y_K = Bv_K + b$.
Then, for any point $x \in K$ with $x = Au + a$, we have
\[
  d(F, x)^2 = \sum_{i = 1}^k \left(1-\sigma_i^2\right)(u - u_F)_i^2
    + d(F, K)^2,
\]
and for any point $y \in F$ with $y = Bv + b$, we have
\[
  d(y, K)^2 = \sum_{i = 1}^k \left(1-\sigma_i^2\right)(v - v_K)_i^2
    + d(F, K)^2.
\]
\end{lemma}
\begin{proof}
  We show the calculation for $d(F, x)$. The calculation for
  $d(y, K)$ is symmetric. Let $x \in K$ with $x = Au + a$
  be given. Let $y_x \in F$ be the projection of $x$ onto $F$.
  Then,
  \[
    d(F, x)^2
              = \|x - y_x \|^2 =
              \| (x - x_F) + (x_F - y_K) + (y_K - y_x) \|^2
              = \| (x - x_F) - (y_x - y_K) \|^2
                + \| x_F - y_K \|^2,
  \]
  where the last equality is due to Pythagoras, since
  $x - x_F$ lies in $K$, $y_x - y_K$ lies in $F$, and
  $x_F - y_K$ is orthogonal to both $K$ and $F$.
  Now,  we have
  $y_x = BB^Tx + b$. Similarly, since $y_K$ is
  the projection of $x_F$ onto $F$, we have
  $y_K = BB^Tx_F + b$. Thus,
  \[
    d(F, x)^2
              = \left\| (x - x_F) - BB^{T}(x - x_F) \right\|^2
                + d(F, K)^2 =
  \left\|\left(A - BB^TA\right)(u - u_F) \right\|^2
                + d(F, K)^2,
  \]
  using the definition of $x$ and $x_F$.
  By Lemma~\ref{lem:matrices}, the columns
  $a_1^\perp, \dots, a_k^\perp$ of the matrix
  $A - BB^TA'$ are pairwise orthogonal and for
  $i = 1, \dots, k$, we have $\| a_i^\perp \|^2 = 1-\sigma_i^2$.
  Pythagoras gives
  \[
    d(F, x)^2 =
                \sum_{i = 1}^k \|a_i^\perp\|^2
                  (u - u_F)_i^2 + d(F, K)^2
    = \sum_{i = 1}^k \left(1 - \sigma_i^2\right)
    (u - u_F)_i^2 + d(F, K)^2.
  \]
\end{proof}

\begin{alg}
\SetKwInOut{Input}{Input}
\SetKwFunction{CreatePartitionStructure}{CreatePartitionStructure}
\SetKwFunction{CreateSlabStructure}{CreateSlabStructure}
\Input{query $k$-flat $F \subseteq \R^d$;
  an estimate $\widetilde{r}$ with
  $d(F, Q) \in [\widetilde{r}/n^t, \widetilde{r}]$.}

  $M \gets A'^TB'$.

  Compute an SVD $M = U \Sigma V^T$ of $M$ with singular values
    $1 \geq \sigma_1 \geq \dots \geq \sigma_k \geq 0$.

  \If{$\sigma_k = 1$} {

    $f \gets $projection of $F$ onto $K^\perp$
    \tcc*{$F$ and $K$ are parallel; $f$ is a point}

    $r \gets$  $c'$-ANN for $f$ in $Q_b$

     \Return $r$
  }

  Reparametrize $K$ according to $U$ and $F$ according to $V$.

 \tcc{Near case}
  $\mathcal{G} \gets $ set of approximate patches obtained
    by combining Lemma~\ref{lem:tflats} and~\ref{lem:patches}

  $R \gets \emptyset$

  \ForEach{$G \gets \mathcal{G}$} {
    $R \gets R \cup \text{result of approximate nearest-neighbor
      query for $G$ as in Lemma~\ref{lem:patch_NN}}$
  }
 \tcc{Far case}
    $R \gets R \cup \text{result of approximate nearest-neighbor
    for $G$ as in Lemma~\ref{lem:farQuery}}$

  \Return point in $R$ that minimizes the distance to $F$
 \caption{QueryClusterStructure}
 \label{alg:queryCluster}
\end{alg}

We now give a brief overview of the query algorithm,
refer to Algorithm~\ref{alg:queryCluster} for pseudocode.
First,  we check for the special case that $F$ and $K$ are
parallel, i.e., that $\sigma_1 = \dots = \sigma_k = 1$. In
this case, we need to perform only a single $c'$-ANN
query in $Q_b$ to obtain the desired result.
If $F$ and $K$ are not parallel, we distinguish two
scenarios: if $F$ is far from $Q$, we can approximate $Q$ by
its projection $Q_a$ onto $K$. Thus, we take the closest point
$x_F$ in $K$ to $F$, and we return an approximate nearest neighbor for
$x_F$ in $Q_a$ according to an appropriate metric derived
from Lemma~\ref{lem:svddistance}. Details can be found in
Section~\ref{sec:far}.
If $F$ is close to $Q$, we use Lemma~\ref{lem:svddistance} to
discretize the relevant
part of $F$ into \emph{patches}, such that each patch is
parallel to $K$ and such that the best nearest
neighbor in $Q$ for the patches provides an approximate nearest
neighbor for $F$. Each patch can then
be handled essentially by an appropriate nearest neighbor
query in $K^\perp$. Details follow in Section~\ref{sec:close}.
 We say $F$ and $Q$ are \emph{close}
if $d(F, Q) \leq \alpha/\eps$, and \emph{far} if
$d(F, Q) > \alpha/\eps$. Recall that we chose $\eps = 1/100\log n$.

\subsubsection{Near: $d(F, Q) \leq \alpha/\eps$}\label{sec:close}

We use our reparametrization of $F$ and $K$ to split the coordinates
as follows: recall that
$1 \geq \sigma_1 \geq \dots \geq \sigma_k \geq 0$ are the singular
values of $M = A'^TB'$. Pick $l \in \{0, \dots, k\}$ such that
$1 \geq \sigma_i \geq \sqrt{1-\eps}$, for $i = 1, \dots, l$, and
$\sqrt{1-\eps} > \sigma_i \geq 0$, for $i = l+1, \dots, k$.
For a $d \times k$ matrix $X$, let
$X_{[i]}$ denote the $d \times i$ matrix with the first $i$
columns of $X$, and $X_{-[i]}$ the $d \times (k-i)$ matrix with
the remaining $k-i$ columns of $X$.
Similarly, for a vector $v \in \R^k$, let
$v_{[i]}$ be the vector in $\R^i$ with the  first $i$
coordinates of $v$, and $v_{-[i]}$ the vector in $\R^{k-i}$ with the
remaining $k-i$ coordinates of $v$.

The following lemma is an immediate consequence of
Lemma~\ref{lem:svddistance}. It tells us that
we can partition the directions in $F$ into those that are
almost parallel to $K$ and those that are almost orthogonal
to $K$. Along the orthogonal directions, we discretize
$F$ into few lower-dimensional flats that are almost
parallel to $K$. After that, we approximate these flats
by few patches that are actually parallel to $K$.
These patches are then used to perform the query.

\begin{lemma}\label{lem:orthogDiag}
  Let $y \in F$ be a point and
  $y_K \in F$ with  $d(F, K) = d(y_K, K)$.
  Write $y_K = Bv_K + b$ and $y = Bv + b$.
  Then, $\left\|(v - v_K)_{-[l]}\right\| \leq d(y, K)/\sqrt{\eps}$.
\end{lemma}
\begin{proof}
  By Lemma~\ref{lem:svddistance} and the choice of $l$,
  \[
    d(y, K)^2 =
    \sum_{i = 1}^k \left(1-\sigma_i^2\right) (v - v_K)_i^2 + d(F, K)^2
    \geq
    \sum_{i = l+1}^k \left(1-\sigma_i^2\right) (v - v_K)_i^2
    \geq
      \eps \left\|(v - v_K)_{-[l]}\right\|^2.
    \]
\end{proof}

Using Lemma~\ref{lem:orthogDiag}, we can discretize
the query $F$ into a set of  $l$-flats that are
almost parallel to the cluster flat $K$.
\begin{lemma}\label{lem:tflats}
  There is a set $\mathcal{L}$  of
  $O((n^{2t}k^{1/2}\eps^{-5/2})^{k-l})$
  $l$-flats such
  that the following holds:
  (i) for every $L \in \mathcal{L}$, we have
    $L \subseteq F$;
  (ii) for every $L \in \mathcal{L}$ and for every unit vector
    $u \in L$, the projection
    of $u$ onto $K$ has length at least $\sqrt{1-\eps}$; and
  (iii) if $d(F, Q) \in [\alpha/n^{2t}, \alpha/\eps]$,
  then there is an $l$-flat $L \in \mathcal{L}$ with
  $d(L, Q) \leq (1 + \eps)d(F, Q)$.
\end{lemma}
\begin{proof}
  Let $y_K = Bv_K + b \in F$ be a point in $F$ with
  $d(F, K) = d(y_K, K)$.  Furthermore, let
  \[
    \tau = \frac{\alpha\eps}{n^{2t}\sqrt{k}}
    \quad\text{ and }\quad
    o_\tau = \left\lceil \frac{ n^{2t} \sqrt{k}}{\eps^{5/2}}
    \right\rceil.
  \]
  Using $\tau$ and $\sigma_\tau$, we
  define a set $I$ of \emph{index vectors} with
  $
    I = \left\{-o_\tau \tau, (-o_\tau + 1) \tau, \dots,
    o_\tau \tau \right\}^{k-l}
  $
  and $|I| = O(o_\tau^{k-l}) = O((n^{2t}k^{1/2}\eps^{-5/2})^{k-l})$.
  For each $i \in I$, we define the $l$-flat $L_i$ as
  \[
    L_i: w \mapsto B_{[l]}w + B_{-[l]}\left((v_K)_{-[l]} + i\right)
    + b.
  \]
  Our desired set of approximate query $l$-flats is now $\mathcal{L} = \{L_i \mid
  i \in I\}$.

  The set $\mathcal{L}$ meets properties (i) and (ii)
  by construction, so it remains to verify (iii).
  For this, we take a point $y_Q \in F$ with $d(F, Q) = d(y_Q, Q)$.
  We write $y_Q = Bv_Q + b$,  and we define
  $s = (v_Q - v_K)_{-[l]}$.
  We assumed that
  $d(y_Q, K) \leq \alpha/\eps$,
  so Lemma~\ref{lem:orthogDiag} gives
  $\|s \| \leq \alpha/\eps^{3/2}$.
  It follows that by rounding each coordinate of $s$ to
  the nearest multiple of $\tau$,
  we obtain an index vector $i_Q \in I$ with
  $\| i_Q - s \| \leq \tau\sqrt{k} = \eps\alpha/n^{2t}$.
  Hence, considering the point in $L_{i_Q}$ with $w = (v_Q)_{[l]}$,
  we get
  \begin{align*}
    d(L_{i_Q}, Q) &\leq
    d(L_{i_Q}, y_Q) + d(y_Q, Q)\\
    &\leq  \left\|B_{[l]} (v_Q)_{[l]} +
      B_{-[l]}\left((v_K)_{-[l]}+  i_Q\right) +
      b - Bv_Q - b \right\| + d(F, Q)\\
      &=
      \left\|B_{-[l]}\left((v_K)_{-[l]} + i_Q - (v_Q)_{-[l]}\right)
        \right\| + d(F, Q)\\
      &= \|(v_K - v_Q)_{-[l]} + i_Q\| + d(F, Q)\tag{*}\\
      &= \|i_Q - s\| + d(F, Q)\\
      &\leq \eps\alpha/n^{2t} + d(F, Q)\\
      &\leq (1 + \eps)d(F, Q) \tag{**},
  \end{align*}
  where in (*) we used that the columns of
  $B_{-[l]}$ are orthonormal and in (**) we used the
  assumption $d(F, Q) \geq \alpha/n^{2t}$.
\end{proof}

From now on, we focus on an approximate query $l$-flat
$L: w \mapsto B_1w + b_1$ with $B_1 = B_{[l]}$.
Our next goal is to approximate $L$  by a set of patches such
that each is parallel to $K$.

\begin{lemma}\label{lem:patches}
  There is a set $\mathcal{G}$  of $O((n^{2t}k^{1/2}\eps^{-2})^l)$
  patches such
  that the following holds:
  (i) every $G \in \mathcal{G}$ is an $l$-dimensional
    polytope, given by $O(l)$ inequalities;
  (ii) for every $G \in \mathcal{G}$,
  the affine hull of $G$ is parallel to $K$;
  (iii) if $d(L, Q) \in [\alpha/n^{2t}, 2\alpha/\eps]$,
  then there exists $G \in \mathcal{G}$ with
  $d(G, Q) \leq (1 + \eps)d(L, Q)$;
  (iv) for all $G \in \mathcal{G}$ and for all $q \in Q$, we have
  $d(G, q) \geq (1 - \eps)d(L, q)$.
\end{lemma}
\begin{proof}
Let $C = AA^TB_1$ be the $d \times l$ matrix
whose columns $b_1^\parallel, \dots, b_l^\parallel$ constitute
the projections of the columns of $B$ onto $K$.
By Lemma~\ref{lem:matrices},
the vectors $b_i^\parallel$ are  orthogonal with
$\|b_i^\parallel\| = \sigma_i$, for
$i = 1, \dots, l$, and the columns
$b_1^\perp, \dots, b_l^\perp$
of the matrix $B_1 - C$ also constitute an orthogonal set,
with $\|b_i^\perp\|^2 = 1-\sigma_i^2$,
for $i = 1, \dots, l$. Let $z_K$
be a point in $L$ that minimizes the distance to $K$, and write
$z_K = B_1w_K + b_1$.  Furthermore, let
\[
  \tau_i =
  \frac{\alpha\eps}{n^{2t}\sqrt{l(1-\sigma_i^2)}},
  \quad
  \text{for $i = 1, \dots, l$, and}
  \quad
  o_{\tau} = \left\lceil \frac{2 n^{2t} \sqrt{l}}{\eps^{2}}\right\rceil.
\]
We use the $\tau_i$ and $o_\tau$ to define a set $I$ of
\emph{index vectors} as
$
  I = \prod_{i=1}^{l}
  \{-o_\tau\tau_i, (-o_\tau + 1)\tau_i, \dots,
    o_\tau \tau_i\}
$.
We have $|I| = O(o_\tau^l) = O((n^{2t}k^{1/2}\eps^{-2})^l)$.
For each index vector $i \in I$, we define the patch $G_i$ as
\[
  G_i: w \mapsto Cw + B_1(w_K + i)
  + b_1, \text{ subject to } w \in \prod_{i=1}^{l}
  \left[0,  \tau_i\right].
\]
Our desired set of approximate query patches is now
$\mathcal{G} = \{G_i \mid i \in I\}$.
The set $\mathcal{G}$ fulfills
properties (i) and (ii) by construction, so it remains to check
(iii). Fix a point $z \in L$. Since $L \subseteq F$, we can
write $z =  B_1w + b_1 = Bv + b$, where the vector $w$
represents the coordinates of $z$ in $L$ and the vector $v$
represents the coordinates of $z$ in $F$.
By Lemma~\ref{lem:svddistance},
\[
  d(z, K)^2 = \sum_{i=1}^k (1-\sigma_i^2)(v - v_K)_i^2
           + d(F, K)^2,
\]
where the vector $v_K$ represents the coordinates of a point in
$F$ that is closest to $K$. By definition of $L$, the last
$k-l$ coordinates $v_{-[l]}$ in $F$  are the same for all points
$z \in L$, so we can conclude that the coordinates for
a closest point to $K$ in $L$ are given by
$w_K = (v_K)_{[l]}$ and that
\begin{equation}\label{equ:distLK}
  d(z, K)^2 = \sum_{i=1}^l (1-\sigma_i^2)(w - w_K)_i^2
           + d(L, K)^2.
\end{equation}
Now take a point $z_Q$ in $L$ with $d(z_Q, Q) = d(L, Q)$ and
write $z_Q = B_1w_Q + b_1$.
Since we assumed $d(L, Q) \leq 2\alpha/\eps$, (\ref{equ:distLK})
implies that for $i = 1, \dots, l$, we have
$|(w_Q - w_K)_i| \leq 2\alpha/\Bigl(\eps\sqrt{1 + \sigma_i^2}\Bigr)$.
Thus, if for $i = 1, \dots, l$, we round
$(w_Q - w_K)_i$ down to the next multiple of
$\tau_i$,
we obtain an index vector $i_Q \in I$ with
$(w_Q - w_K) - i_Q \in \prod_{i=1}^{l} \left[0,  \tau_i\right]$.
We set $s_Q = (w_Q - w_K) - i_Q$. Considering the point
$Cs_Q + B_1(u_K + i_Q) + b_1$ in  $G_{i_Q}$, we see that
\begin{multline*}
  d(G_{i_Q}, z_Q)^2
    \leq \|Cs_Q + B_1(w_K + i_Q) + b_1 - B_1w_Q - b_1 \|^2
    = \|Cs_Q - B_1((w_Q - w_K) - i_Q) \|^2\\
    = \|(C - B_1)s_Q\|^2
    = \sum_{i=1}^l (1-\sigma_i^2)(s_Q)_i^2
    \leq \sum_{i=1}^l (1-\sigma_i^2)\tau_i^2
    = \eps^2\alpha^2/n^{4t},
\end{multline*}
using the properties of the matrix $B_1 - C$ stated above.
It follows that
\[
  d(G_{i_Q}, Q) \leq
  d(G_{i_Q}, z_Q) + d(z_Q, Q)
    \leq
      \eps\alpha/n^{2t} + d(L, Q)
    \leq (1+\eps)d(L, Q),
\]
since we  assumed $d(L, Q) \geq \alpha/n^{2t}$. This proves
(iii). Property (iv) is obtained similarly.
Let $G_i \in G$, $q \in Q$ and let
$z$ be a point in $G_i$. Write $z = Cw + B_1(w_K + i) + b_1$,
where $w \in \prod_{i=1}^{t} \left[0,  \sigma_i\right]$.
Considering the point $z_x = B_1(w + w_K + i) + b_1$ in $L$,
we see that
\[
  d(G_i, r_x)^2  \leq \|z - z_x\|^2
    = \|(C - B_1)w\|\leq \eps^2\alpha^2/n^{4t}.
\]
Thus,
\[
  d(G_{i}, q) \geq
    d(z_x, q) - d(G_i, z_x)
    \geq
      d(L, q) - \eps\alpha/n^{2t}
    \geq (1 - \eps)d(L, q).
\]
\end{proof}

Finally, we have a patch $G: w \mapsto Cw + b_2$, and we are
looking for an approximate nearest neighbor for $G$ in $Q$.
The next lemma states how this can be done.
\begin{lemma}\label{lem:patch_NN}
  Suppose that $d(G, Q) \in [\alpha/2n^{2t},
    3\alpha/\eps]$. We can find a point
    $\widetilde{q} \in Q$ with
    $d(G, \widetilde{q}) \leq (1 - 1/2\log n)c d(G, Q)$
    in total time
    $O_c((k^2n^{2t}/\eps^2)(m^{1-1/k+\rho/k} + (d/k)\log m))$.
\end{lemma}
\begin{proof}
  Let $G_a$ be the projection  of $G$ onto $K$, and let
  $g$ be the projection of $G$ onto $K^\perp$.
  Since $G$ and $K$ are parallel,
  $g$ is a point, and
  $G_a$ is of the form $G_a : w \mapsto Cw + a_2$, with
  $a_2 \in K$ and $w \in \prod_{i=1}^t [0, \tau_i]$.
  Let $G_{a}^+ =
  \{x \in K \mid d_\infty(x, G_a) \leq
  3\alpha\sqrt{k}/\eps\}$, where
  $d_\infty(\cdot, \cdot)$ denotes the
  $\ell_\infty$-distance with respect to the coordinate system induced
  by $A$.
  We subdivide the set $G_a^+ \setminus G_a$,
  into a collection $\mathcal{C}$ of
  axis-parallel cubes, each with diameter
  $\eps\alpha/2n^{2t}$.
  The cubes in $\mathcal{C}$
  have side length $\eps\alpha/2n^{2t}\sqrt{k}$,
  the total number of cubes is
  $O((kn^{2t}/\eps^2)^k)$, and the
  boundaries of the cubes lie on
  $O(k^2n^{2t}/\eps^2)$ hyperplanes.

  We now search the partition tree $\mathcal{T}$
  to find the highest nodes $(\Delta, Q)$ in
  $\mathcal{T}$  whose simplices $\Delta$ are completely contained
  in a single cube of $\mathcal{C}$.
  This is done as follows: we begin at the root
  of $\mathcal{T}$, and we check for all children $(\Delta, Q)$
  and for all boundary hyperplanes $h$ of $\mathcal{C}$
  whether the simplex $\Delta$ crosses the boundary $h$.
  If a child $(\Delta, Q)$ crosses no hyperplane, we label it
  with the corresponding cube in $\mathcal{C}$ (or with $G_a$).
  Otherwise, we recurse on $(\Delta, Q)$
  with all the boundary hyperplanes that it crosses.

  In the end, we have obtained a set $\mathcal{D}$ of simplices
  such that each simplex in $\mathcal{D}$ is completely contained in
  a cube of $\mathcal{C}$.
  The total number of simplices
  in $\mathcal{D}$ is $s = O((k^2n^{2t}/\eps^2)m^{1-1/k})$, by
  Theorem~\ref{thm:optpartition}.
  For each simplex in $\mathcal{D}$,
  we query the corresponding $c'$-ANN structure.
  Let $R \subseteq Q_b$ be the set of the query results.
  For each point $q_b \in R$, we take the corresponding point
  $q \in Q$, and we compute the distance $d(q, G)$.
  We return a point $\widetilde{q}$ that minimizes $d(q, G)$.
  The query time is dominated by the time for the ANN queries.
  For each $\Delta \in \mathcal{D}$, let $m_\Delta$ be the number
  of points in the corresponding ANN structure. By
  assumption, an ANN-query takes time
  $O_c(m_\Delta^\rho + d\log m_\Delta)$,
  so the total query time is proportional to
  \begin{multline*}
    \sum_{\Delta \in \mathcal{D}} m_\Delta^\rho + d\log m_\Delta \leq
    s \left( \sum_{\Delta \in \mathcal{D}} m_\Delta/s \right)^\rho  +
    s d \log \left( \sum_{\Delta \in \mathcal{D}} m_\Delta/s \right)\\
    \leq
    O_c\left((k^2n^{2t}/\eps^2)(m^{1-1/k+\rho/k} + (d/k)\log m)\right),
  \end{multline*}
  using the fact that $m \mapsto m^\rho + d\log m$ is concave and
  that $\sum_{\Delta \in \mathcal{D}} m_\Delta \leq m$.

  It remains to prove that approximation bound.
  Take a point $q^*$ in $Q$ with $d(q^*, Q) = d(Q, G)$. Since
  we assumed that $d(Q, G) \leq 3\alpha/\eps$, the projection
  $q^*_a$ of $q^*$ onto $K$ lies in $G_a^+$.
  Let $\Delta^*$ be the simplex in $\mathcal{D}$ with
  $q^*_a \in \Delta^*$. Suppose that the ANN-query for $\Delta^*$
  returns a point $\widehat{q} \in Q$.
  Thus, in $K^\perp$, we have
  $d(\widehat{q_b}, g) \leq c'd(Q_{b\Delta^*}, g) \leq
  c'd(q_b^*, g)$, where $\widehat{q_b}$ and $q^*_b$ are
  the projections of $\widehat{q}$ and $q^*$ onto $K^\perp$ and
  $Q_{b\Delta^*}$ is the point set stored in the ANN-structure of
  $\Delta^*$.
  By the definition of $\mathcal{C}$, in $K$, we have
  $d(\widehat{q_a}, G_a) \leq
  d(q_a^*, G_a) + \eps\alpha/2n^{2t} \leq
  d(q_a^*, G_a) + \eps d(q^*, G)$,
  where $\widehat{q_a}$ is the projection of $\widehat{q}$ onto $K$.
  By Pythagoras,
  \begin{align*}
    d(\widehat{q}, G)^2 &=
      d(\widehat{q_b}, g)^2 + d(\widehat{q_a}, G_a)^2\\
    &\leq
      c'^2d(q_b^*, g)^2
  + (d(q_a^*, G_a) + \eps d(q^*, G))^2\\
    &\leq
      c'^2d(q_b^*, g)^2
  + d(q_a^*, G_a)^2  + (2\eps + \eps^2)d(q^*, G)^2 \\
  &\leq
    (c'^2 + 3\eps)(q^*, G)^2\\
  &\leq
    \left((1-1/\log n)^2c^2 + 3/100\log n\right)(q^*, G)^2\\
  &\leq
    (1-1/2\log n)^2c^2(q^*, G)^2,
  \end{align*}
  recalling that $c' = (1-1/\log n)c$ and $\eps =
  1/100\log n$.
  Since $d(\widetilde{q}, G) \leq d(\widehat{q}, G)$,
  the result follows.
\end{proof}
\noindent
Of all the candidate points obtained through querying
patches, we return the one closest to $F$.
The following lemma summarizes the properties of the
query algorithm.
\begin{lemma}\label{lem:nearQuery}
  Suppose that $d(F, Q) \in [\alpha/n^{2t}, \alpha/\eps]$.
  Then the query procedure returns
  a point $\widetilde{q} \in Q$ with
  $d(F, \widetilde{q}) \leq cd(F, Q)$ in
  total time
  $O_c((k^2n^{2t}\eps^{-5/2})^{k+1}(m^{1-1/k+\rho/k} + (d/k)\log m))$.
\end{lemma}
\begin{proof}
  By Lemmas~\ref{lem:tflats} and~\ref{lem:patches},
  there exists a patch $G$ with $d(G, Q) \leq (1 + \eps)^2d(F, Q)$.
  For this patch, the algorithm from Lemma~\ref{lem:patch_NN}
  returns a point $\widehat{q}$ with
  $d(\widehat{q}, G) \leq (1+1/2\log n)cd(G, Q)$.
  Thus, using Lemma~\ref{lem:patches}(iv), we have
  \[
      (1-\eps)d(\widehat{q}, L) \leq d(\widehat{q}, G)
      \leq (1-1/2\log n)c(1 + \eps)^2d(F, Q)
  \]
  and
  by our choice of $\eps = 1/100\log n$, we get
    \begin{multline*}
    (1-1/2\log n)(1 + \eps)^2/(1-\eps)
      \leq
      (1-1/2\log n)(1 + 3\eps)(1+2\eps)\\
      \leq (1-1/2\log n)(1 + 6/100\log n)
      \leq 1.
    \end{multline*}
\end{proof}

\subsubsection{Far: $d(F, Q) \geq \alpha/\eps$}\label{sec:far}

If $d(F, Q) \geq \alpha/\eps$, we can approximate $Q$ by its
projection $Q_a$ onto $K$ without losing too much.
Thus, we can perform the whole algorithm in $K$.
This is done by a procedure similar to Lemma~\ref{lem:patch_NN}.

\begin{lemma}\label{lem:farNNinQa}
  Suppose we are given an estimate $\widetilde{r}$
  with $d(F, Q_a) \in [\widetilde{r}/2n^{t}, 2\widetilde{r}]$.
  Then, we can find a point $\widetilde{q} \in Q_a$ with
  $d(F, \widetilde{q}) \leq (1 + \eps)d(F, Q_a)$ in
  time $O((k^{3/2}n^{t}/\eps)m^{1-1/k})$.
\end{lemma}
\begin{proof}
Let $x_F$ be a point in $K$ with
$d(F, K) = d(F, x_F)$. Write $x_F = Au_F + a$.
Define
\[
  C = \prod_{i = 1}^k \left((u_F)_i  +
    \left[0,
  2\widetilde{r}/\sqrt{1-\sigma_i^2}\right]\right)
\]
If we take a point $x \in K$ with
$d(x, F) \in [\widetilde{r}/2n^{t}, 2\widetilde{r}]$ and write
$x = Au + a$, then Lemma~\ref{lem:svddistance} gives
\[
  d(F, x)^2 =
    \sum_{i = 1}^k(1-\sigma_i^2)(u - u_F)_i^2 +
         d(F, K)^2,
\]
so  $u \in C$.
We subdivide $C$ into copies of the
hyperrectangle
$\prod_{i = 1}^k[0, \eps\widetilde{r}/2n^t\sqrt{k(1-\sigma_i^2)}]$.
Let $\mathcal{C}$ be the resulting set of hyperrectangles.
The boundaries of the hyperrectangles in $\mathcal{C}$ lie on
$O(k^{3/2}n^{t}/\eps)$ hyperplanes.
We now search the partition tree $\mathcal{T}$
in order to find the highest nodes $(\Delta, Q)$ in
$\mathcal{T}$  whose simplices $\Delta$ are completely contained
in a single hyperrectangle of $\mathcal{C}$.
This is done in the same way as in Lemma~\ref{lem:patch_NN}.

This gives a set $\mathcal{D}$ of simplices
such that each simplex in $\mathcal{D}$ is completely contained in
a hyperrectangle of $\mathcal{C}$.
The total number of simplices
in $\mathcal{D}$ is $O((k^{3/2}n^{t}/\eps)m^{1-1/k})$, by
Theorem~\ref{thm:optpartition}.
For each simplex $\Delta \in \mathcal{D}$,
we pick an arbitrary point $q \in Q_a$ that lies in
$\Delta$, and we compute $d(F, q)$. We return
the point $\widetilde{q} \in Q_a$ that minimizes the
distance to $F$.
The total query time is $O((k^{3/2}n^{t}/\eps)m^{1-1/k})$.

Now let $q^*$ be a point in $Q_a$ with $d(F, Q_a) = d(F, q^*)$,
and let $\Delta^*$ be the simplex $\mathcal{D}$ that contains $q^*$.
Furthermore, let $\widehat{q} \in Q_a$ be the point that the algorithm
examines in $\Delta^*$. Write $q^* = Au^* + a$ and
$\widehat{q} = A\widehat{u} + a$. Since
$q^*$ and $\widehat{q}$ lie in the same hyperrectangle and by
Lemma~\ref{lem:svddistance},
  \begin{multline*}
  d(F, \widehat{q})^2 =
    \sum_{i = 1}^k(1-\sigma_i^2)(\widehat{u} - u_F)_i^2 + d(F, K)^2
  \leq\\ \sum_{i = 1}^k(1-\sigma_i^2)(u^* - u_F)_i^2 +
  \eps^2\widetilde{r}^2/4n^{2t} +
  d(F, K)^2 \leq
   (1+\eps)^2d(F, q^*)^2.
\end{multline*}
Since $d(F, \widetilde{q}) \leq d(F, \widehat{q})$, the result
follows.
\end{proof}

\begin{lemma}\label{lem:farQuery}
  Suppose we are given an estimate $\widetilde{r}$
  with $d(F, Q) \in [\widetilde{r}/n^{t}, \widetilde{r}]$.
  Suppose further that $d(F, Q) \geq \alpha/\eps$.
  Then we can find a $\widetilde{q} \in Q$ with
  $d(F, \widetilde{q}) \leq cd(F, Q)$ in
  time $O((k^{3/2}n^{2t}/\eps)m^{1-1/k})$.
\end{lemma}
\begin{proof}
  For any point $q \in Q$, let $q_a \in Q$ be its projection
  onto $K$. Then, $d(q_a, q) \leq \alpha \leq \eps d(F, Q)$.
  Thus, $d(F, Q_a) \in [(1-\eps)d(F, Q), (1+\eps)d(F, Q)]$,
  and we can apply Lemma~\ref{lem:farNNinQa}. Let
  $\widetilde{q_a} \in Q_a$ be the result of this query, and let
  $\widetilde{q}$ be the corresponding point in $Q$. We have
    \begin{multline*}
    d(F, \widetilde{q}) \leq d(\widetilde{q},
      \widetilde{q_a}) + d(F, \widetilde{q_a})
      \leq \eps d(F, Q) + (1 + \eps)d(F, Q_a)\\
      \leq \eps d(F, Q)+ (1 + \eps)^2d(F, Q)
      \leq (1 + 4\eps) d(F, Q)
      \leq c d(F, Q),
  \end{multline*}
  by our choice of $\eps$.
\end{proof}

By combining Lemmas~\ref{lem:clusterPrep},~\ref{lem:nearQuery},
and~\ref{lem:farQuery}, we obtain Theorem~\ref{thm:q2}.

\section{Approximate $k$-flat Range Reporting in Low Dimensions}
\label{sec:lowdimstructure}

In this section, we present  a data structure for low dimensional
$k$-flat approximate near neighbor reporting.
In Section~\ref{sec:projectionstructures}, we will use it as
a foundation for our projection structures.
The details of the structure are
summarized in Theorem~\ref{thm:lowdimstructure}.
Throughout this section,
we will think of $d$ as a constant, and we will suppress
factors depending on $d$ in the $O$-notation.

\begin{theorem} \label{thm:lowdimstructure}
  Let $P \subset \R^d$ be an $n$-point set.
  We can preprocess $P$ into an $O(n\log^{d-k-1}n)$ space data
  structure for approximate $k$-flat near neighbor queries: given
  a $k$-flat $F$ and a parameter $\alpha$, find a set
  $R \subseteq P$ that contains all $p \in P$ with $d(p,F)
  \leq \alpha$ and no $p \in P$ with
  $d(p,F) > (\lowdimapprox) \alpha$. The
  query time is $O(n^{k/(k+1)}\log^{d-k-1} n + |R|)$.
\end{theorem}

\subsection{Preprocessing}
\label{sec:lowdimpreproc}

Let $E \subset \R^d$ be the $(k+1)$-dimensional subspace of $\R^d$
spanned by the first $k+1$ coordinates, and let $Q$ be the
projection of $P$ onto $E$.\footnote{We assume general position:
any two distinct points in $P$ have distinct projections in $Q$.}
We build a $(k+1)$-dimensional partition tree $\mathcal{T}$
for $Q$, as in Theorem~\ref{thm:optpartition}.
If $d > k + 1$, we also build a \emph{slab structure}
for each node of $\mathcal{T}$.
Let $v$ be such a node, and let $\Xi$ be the simplicial partition
for the children of $v$.
Let $w > 0$. A \emph{$w$-slab} $S$ is a closed region in $E$ that is
bounded by two parallel hyperplanes of distance $w$.
The \emph{median hyperplane} $\widehat{h}$ of $S$ is the
hyperplane inside $S$ that is parallel to the
two boundary hyperplanes and has distance $w/2$ from both.
A $w$-slab $S$ is \emph{full} if there are at least $r^{2/3}$
simplices $\Delta$ in $\Xi$ with $\Delta \subset S$.

\begin{alg}
\SetKwInOut{Input}{Input}
\SetKwFunction{CreateSearchStructure}{CreateSearchStructure}
\SetKwFunction{CreateSlabStructure}{CreateSlabStructure}
\Input{point set $P \subset \R^d$}
\If{$|P| = O(1)$} {
Store $P$ in a list and \Return.
}
$Q \gets$ projection of $P$ onto the subspace $E$ spanned by the
first $k+1$ coordinates.

$\mathcal{T} \gets $ $(k+1)$-dimensional partition tree for $Q$ as
in Theorem~\ref{thm:optpartition}.

\If{$d > k+1$} {
  \ForEach{node $v \in \mathcal{T}$}{
$\Xi_1 \gets $ simplicial partition for the children of $v$

\For{$j \gets  1$ \KwTo $\lfloor r^{1/3}\rfloor$} {

 $D_j \gets$ \CreateSlabStructure{$\Xi_j$}

 $\Xi_{j+1} \gets \Xi_j$ without all simplices inside the
 slab for $D_j$\\
 }
  }
}
 \caption{CreateSearchStructure}
 \label{alg:createpartitionstructure}
\end{alg}

\begin{alg}
  \DontPrintSemicolon
  \SetKwInOut{Input}{Input}
  \SetKwFunction{CreateSearchStructure}{CreateSearchStructure}
  \SetKwFunction{CreateSlabStructure}{CreateSlabStructure}
  \Input{$\Xi_j = (Q_1,\Delta_1),\dots,(Q_{r'},\Delta_{r'})$}

  $V_j \gets$ vertices of the simplices in $\Xi_j$

  For each $(k+1)$-subset $V \subset V_j$, find the
  smallest $w_V > 0$ such that the $w_V$-slab with median
  hyperplane $\aff(V)$ is full.

  Let $w_j$ be the smallest $w_V$; let $S_j$ be the
  corresponding full $w_j$-slab and $\widehat{h}_j = \aff(V)$ its median
  hyperplane.

  Find the set $\mathcal{D}_{j}$ of $r^{2/3}$ simplices
  in $S_j$; let
  $\mathcal{Q}_{j} \gets \bigcup_{\Delta_i \in \mathcal{D}_{j}} Q_i$
  and let $\mathcal{P}_{j}$ be the $d$-dimensional point set
  corresponding to $\mathcal{Q}_{j}$.

  $h_j \gets$ the hyperplane orthogonal to $E$ through $\widehat{h}_j$

  $P' \gets $ projection of $\mathcal{P}_j$ onto $h_j$
  \tcc*[r]{$P'$ is $(d-1)$-dimensional}

  \CreateSearchStructure{$P'$}
  \caption{CreateSlabStructure}
  \label{alg:createslabstructure}
\end{alg}

The slab structure for $v$ is constructed in several iterations.
In iteration $j$, we have a current subset $\Xi_j \subseteq \Xi$
of pairs in the simplicial partition.
For each $(k+1)$-set $v_0, \dots, v_k$ of vertices
of simplices in $\Xi_j$, we determine the smallest
width of a full slab whose median hyperplane is
spanned by $v_0,\dots,v_k$. Let $S_j$ be the smallest among those slabs,
and let $\widehat{h}_j$ be its median hyperplane.
Let $\mathcal{D}_j$ be the $r^{2/3}$ simplices that lie completely
in $S_j$.  We remove $\mathcal{D}_{j}$ and the corresponding point
set $\mathcal{Q}_{j}=\bigcup_{\Delta_i \in \mathcal{D}_{j}} Q_i$
from $\Xi_j$ to obtain $\Xi_{j+1}$. Let
$\mathcal{P}_{j} \subseteq P$ be the
$d$-dimensional point set corresponding to $\mathcal{Q}_{j}$.
We project $\mathcal{P}_{j}$ onto the $d$-dimensional
hyperplane $h_j$ that is orthogonal to $E$ and goes through
$\widehat{h}_j$. We recursively
build a search structure for the $(d-1)$-dimensional projected
point set. The $j$th slab structure $D_j$ at $v$ consists of this
search structure, the hyperplane
$h_j$, and the width $w_j$. This process is repeated until less than
$r^{2/3}$ simplices remain; see
Algorithms~\ref{alg:createpartitionstructure}
and~\ref{alg:createslabstructure} for details.

Denote by $S(n,d)$ the space for a $d$-dimensional search
structure with $n$ points.
The partition tree $\mathcal{T}$ has $O(n)$ nodes, so the
overhead for storing the slabs and partitions is linear.
Thus,
\[
  S(n, d) = O(n) + \sum_{D} S(n_D, d - 1),
\]
where the sum is over all slab structures $D$ and where $n_D$ is the
number of points in the slab structure $D$. Since every
point appears in $O(\log n)$ slab structures, and since the
recursion stops for $d = k+1$, we get
\begin{lemma}
 \label{lem:space}
 The search structure for $n$ points in $d$ dimensions needs space
$O(n\log^{d-k-1} n)$.\qed
\end{lemma}

\subsection{Processing a Query}

For a query, we are given a distance threshold $\alpha > 0$ and a
$k$-flat $F$. For the recursion, we will need to query
the search structure with a $k$-dimensional
polytope.  We obtain the initial query polytope by
intersecting the flat $F$ with the bounding box of $P$
extended by $\alpha$ in each direction. With  slight abuse of
notation, we still call this polytope $F$.

A query for $F$ and $\alpha$ is processed by using the
slab structures for small enough slabs and by recursing in
the partition tree for the remaining points. Details follow.

Suppose we are at some node $v$ of the partition tree, and
let $j^*$ be the largest integer
with $w_{j^*} \leq (4k+2)\alpha$. For $j = 1, \dots, j^*$, we
recursively query each slab structure $D_j$ as follows: let
$\widetilde{F} \subseteq F$ be the polytope containing the points
in $F$ with distance at most $\alpha + w_j/2$
from $h_j$, and let $F_h$ be the projection of $\widetilde{F}$
onto $h_j$.  We
query the search structure in $D_j$ with $F_h$ and $\alpha$.
Next, we project $F$ onto the
subspace $E$ spanned by the first $k+1$ coordinates.
Let $\mathcal{D}$ be the simplices in $\Xi_{j^*+1}$ with
distance at most $\alpha$ from the projection.  For each simplex
in $\mathcal{D}$, we recursively query the corresponding child in
the partition tree. Upon reaching the bottom of the recursion
(i.e., $|P| = O(1)$), we collect all points within distance
$\alpha$ from $F$ in the set $R$.

\begin{alg}
 \SetKwInOut{Input}{Input}
 \SetKwFunction{Query}{query}
 \SetKwInOut{Output}{Output}
 \Input{polytope $F$, distance threshold $\alpha > 0$}
 \Output{point set $R \subseteq P$}
 $R \gets \emptyset$

 \If{$|P| = O(1)$}{
   $R \gets \{ p \in P \mid d(p, F) \leq \alpha\}$
 } \ElseIf{$d = k+1$ } {
  \label{alg:querylinesegment:d2start}
 Compute polytope $F_{\diamond}$ as described.
 \label{querylinesegment:computerectangle}

 $R \gets R \cup \text{all points of $P$ inside $F_{\diamond}$}$
 \label{alg:querylinesegment:d2end}

 }
 \Else{
  $j^{*} \gets$ the largest integer with
  $w_{j^{*}} \leq (4k+2)\alpha$

 \For{$j \gets 1$ \KwTo $j^*$}{
  \label{alg:querylinesegment:slabqstart}
  $F_h \gets $ projection of $\widetilde{F}$ onto $h_j$ as described

  $R \gets R \cup D_j$.\Query{$F_h$, $\alpha$}
  \label{alg:querylinesegment:queryslabrecursively}
  }

  $\widehat{F} \gets $ projection of $F$ onto the subspace $E$
  spanned by the
  first $k+1$ coordinates

 $\mathcal{D} \gets $ simplices in $\Xi_{j^* + 1}$

  $\mathcal{D}' \gets \{\Delta \in \mathcal{D} \mid
  d(\Delta, \widehat{F}) \leq \alpha\}$

   \ForEach{$\Delta \in \mathcal{D}'$}{
     $R \gets R \cup \text{result of recursive query to partition
       tree node for $\Delta$.}$
  \label{alg:querytrianglerecursive}
    }
 }
 \Return{$R$}\label{alg:querylinesegment:return}
\caption{Find a superset $R$ of all points in $P$ with distance
less than $\alpha$ from a query polytope $F$.}
 \label{alg:querylinesegment}
\end{alg}

If $d=k+1$, we approximate the region of interest by the
polytope
$F_{\diamond} = \{x \in \R^d \mid d_1(x, F) \leq \alpha\}$,
where $d_1(\cdot, \cdot)$ denotes the $\ell_1$-metric
in $\R^d$.
Then, we query the partition tree $\mathcal{T}$ to find all points
of $P$ that lie inside $F_{\diamond}$.
We prove in Lemma~\ref{lem:querytime} that $F_{\diamond}$ is
a polytope with
$O(d^{O(k^2)})$ facets;
see Algorithm~\ref{alg:querylinesegment} for details.
The following two lemmas analyze the correctness and
query time of the algorithm.

\begin{lemma}
\label{lem:distancerror}
 The set $R$ contains all $p \in P$ with  $d(p,F) \leq \alpha$ and
 no $p \in P$ with $d(p,F) > \kappa\alpha$, where
 $\kappa=\lowdimapprox$.
\end{lemma}
\begin{proof}
  The proof is  by induction on the size $n$ of $P$ and on
  the dimension $d$.
  If $n = O(1)$, we return all points with distance at most
  $\alpha$ to $F$. If $d=k+1$, we report the
  points inside the polytope
  $F_\diamond$
  (lines~\ref{alg:querylinesegment:d2start}--\ref{alg:querylinesegment:d2end})
  using $\mathcal{T}$.
  Since $\|x\|_2 \leq \|x\|_1 \leq \sqrt{k+1}\|x\|_2$
  holds for all $x\in \R^{k+1}$, the polytope $F_\diamond$ contains
  all points with distance at most $\alpha$ from $F$ and no point
  with distance more than $\alpha\sqrt{k+1}$ from $F$. Thus,
  correctness also follows in this case.

  \begin{figure}[htbp]
    \centering
    \includegraphics[scale=0.6]{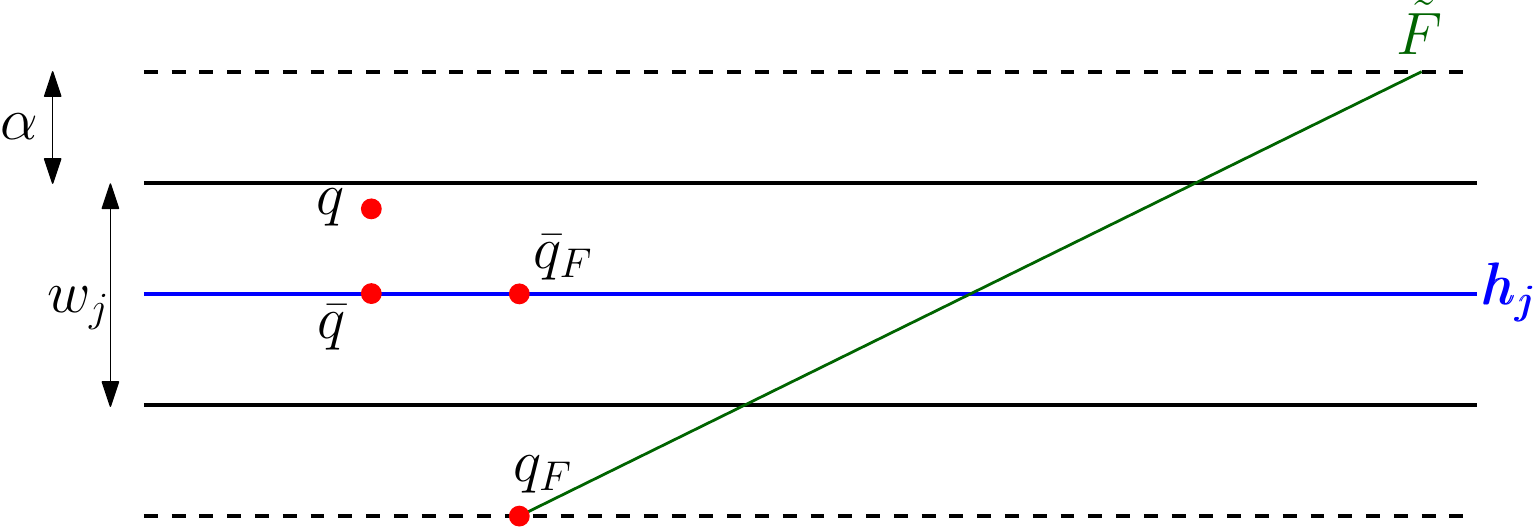}
    \caption{The distance error due to the reduction of the dimension
      in the slab structure $D_j$.}
    \label{fig:distanceerror}
  \end{figure}

  In the general case  ($d > k+1$ and $n$ not constant), we prove
  the lemma individually for the slab structures and for the
  partition tree.  Let $D_j$ be a slab
  structure and $P_j$ the corresponding $d$-dimensional point set.
  Fix some point $p \in P_j$ with $d(p, F) \leq \alpha$. To
  query $D_j$, we take the subpolytope $\widetilde{F} \subseteq F$
  with distance at most $\alpha + w_j/2$ from the
  median hyperplane $h_j$, and we project it onto $h_j$. Let
  $F_h$ be this projection.
  Since
  orthogonal projections can only
  decrease distances, we have
  $p \in D_j$\texttt{.query(}$F_h$, $\alpha$\texttt{)} by induction.
  Now fix a point $q \in D_j$\texttt{.query(}$F_h$, $\alpha$\texttt{)}.
  We must argue that $d(q, F_h) \leq \lowdimapprox$.
  Let $\bar{q}\in\R^{d-1}$ be the projection of $q$ onto $h_j$
  and $\bar{q}_{F} \in F_h$ the closest point to
  $\bar{q}$ in $F_h$.
  Let $q_F \in \widetilde{F}$ be the corresponding $d$-dimensional
  point (see Fig.~\ref{fig:distanceerror}).
  By triangle inequality and the induction hypothesis,
  \begin{align*}
      d(q,F) &=
        d(q,\widetilde{F}) \leq d(q, \bar{q}) + d(\bar{q}, \bar{q_F}) +
         d(\bar{q_F}, q_F) \\
      & \leq w_j/2 + ((4k+3) (d-k-2) + \sqrt{k+1}) \alpha +
         (\alpha + w_j/2).
  \end{align*}
  By construction, we have $w_j \leq (4k+2)\alpha$, so
  $d(q,F) \leq (\lowdimapprox)\alpha$, as claimed.

  Consider now a child in the partition tree queried in
  line~\ref{alg:querytrianglerecursive}, and let $P_j$ be the
  corresponding $d$-dimensional point set. Since
  $|P_j|  < |P|$, the claim follows by induction.
\end{proof}

\begin{lemma}
  \label{lem:querytime}
  The query time is $O(n^{k/(k+1)}\log^{d-k-1} n + |R|)$.
\end{lemma}
\begin{proof}
  Let $Q(n,d)$ be the total query time.

  First, let $d > k+1$.
  We bound the time to query the partition tree $\mathcal{T}$.
  Let $\widehat{F}$ be the projection of $F$ onto $E$.
  Furthermore, let $V$ be the set of nodes in $\mathcal{T}$ that
  are visited during a query, and let $\mathcal{D}$ be the
  corresponding
  simplices. By construction, all simplices in $\mathcal{D}$
  have distance at most $\alpha$ from $\widehat{F}$.
  Consider the $2\alpha$-slab $S$ whose
  median hyperplane contains $\widehat{F}$.
  We partition $V$ into two sets: the nodes
  $V_B$ whose simplices intersect $\partial S$, and the nodes
  $V_C$ whose simplices lie completely in $S$.
  First, since the simplex for each node in $\mathcal{T}$ is
  contained in the simplex for its parent node, we observe that
  $V_B$ constitutes a connected subtree of $\mathcal{T}$, starting
  at the root.  The nodes of $V_C$ form several connected subtrees,
  each hanging off a node in $V_B$. Furthermore, by construction,
  each node from $V$ has at most $r^{1/3}$ children from $V_B$.
  Let $V_\ell$ be the set of nodes in $V$ with level $\ell$,
  for $\ell = 0, \dots, \log_r n$,
  and let $m_\ell = |V_\ell|$.
  By Theorem~\ref{thm:optpartition}, we have
  $|V_\ell \cap V_B| \leq
  O(r^{\ell k/(k+1)} + r^{(k-1)/k}m_{\ell - 1} +
  r\ell \log r \log n)$. Since
  $|V_\ell \cap V_C| \leq r^{1/3} m_{\ell - 1}$,
  we get
  \[
    m_\ell = |V_\ell| =
    O\left(r^{\ell k/(k+1)} + (r^{(k-1)/k} + r^{1/3})m_{\ell-1} +
    r\ell \log r \log n\right).
  \]
  For any $\delta > \max(0, 1/3-(k-1)/k)$, if we choose
  $r$ large enough depending on $\delta$, this solves to
  $m_\ell = O(r^{\ell k/(k+1)} +
  r^{\ell ((k-1)/k + \delta)}\log n)$.
  Thus, we get
  \begin{equation}\label{equ:qrec}
    Q(n, d) =
    \sum_{\ell = 0}^{\log_r n} O(r^{\ell k/(k+1)} +
  r^{\ell ((k-1)/k + \delta)}\log n)(O(r)  + Q(n/r^{\ell}, d-1)).
  \end{equation}

  \newcommand{\magicnumber}{((2k+2)(5d_0 - 2k)^{k/2})^{(k+1)/2}}
  For $d = k+1$, we use $\mathcal{T}$ directly.
  Thus, by Theorem~\ref{thm:optpartition}, the query time $Q(n, k+1)$
  is $O(f_{k+1}n^{1-1/k} + |R_F|)$, where
  $f_{k+1}$ is the number of facets of $F_\Diamond$ and
  $R_F$ is the answer set.
  We claim that
  $f_{k+1}$ is bounded by $\magicnumber$.
  Recall that $F_{\diamond}$ is the Minkowski sum of $F$ and
  the $\ell_1$-ball with radius
  $\alpha$ as in Algorithm~\ref{alg:querylinesegment}
  line~\ref{querylinesegment:computerectangle}.
  Initially $F$ is the intersection of the query $k$-flat with the
  extended bounding box of $P$. This intersection can
  be described by at most
  $d_0-k+2d_0=3d_0-k$ oriented half-spaces, where $d_0$ denotes
  the initial dimension. In each recursive
  step, we intersect $F$ with the $2$ bounding hyperplanes of a slab.
  Therefore, the descriptive complexity of $F$ in the base case is
  at most $5d_0-2k$. By duality and the Upper Bound
  theorem~\cite{Matousek02}, the $\mathcal{V}$-description of
  $F$ consists of at most $(5d_0 - 2k)^{k/2}$ vertices. Using
  that the Minkowski sum of two
  polytopes with $v_1$ and $v_2$ vertices has at most
  $v_1v_2$ vertices, we deduce that $F_\diamond$ has at
  most $(2k+2)(5d_0 - 2k)^{k/2}$ vertices.
  Applying the upper bound theorem again, it follows that
  $f_{k+1} = \magicnumber$, as claimed.

  Thus, plugging the base case into (\ref{equ:qrec}),
  we get that the overall query time $Q(n,d)$ is  bounded by
  $O(n^{k/(k+1)}\log^{d-k-1} + |R|)$.
\end{proof}

Theorem~\ref{thm:lowdimstructure} follows immediately from
Lemmas~\ref{lem:space}, \ref{lem:distancerror}, and
\ref{lem:querytime}.

\subsection{Approximate $k$-Flat Nearest Neighbor Queries}
\label{sec:approximateNNlowdim}
We now show how to extend our
data structure from Section~\ref{sec:lowdimpreproc}
for approximate $k$-flat nearest neighbor queries with
multiplicative error $\lowdimapprox$.
That is, given an $n$-point set
$P \subset \R^d$, we want to find for any
given query flat $F \subset \R^d$
a point $p \in P$ with $d(p,F) \leq (\lowdimapprox) d(P,F)$.
We reduce this problem to a near neighbor query by choosing an
appropriate
threshold $\alpha$ that ensures
$|R| = O(\sqrt{n})$, using random sampling.
For preprocessing we build the data structure $D$ from
Theorem~\ref{thm:lowdimstructure} for $P$.

Let a query flat $F$ be given.
The \emph{$F$-rank} of a point $p \in P$ is the number of
points in $P$ that are closer to $F$ than $p$.
Let $X \subseteq P$ be a random sample obtained by
taking each point in $P$ independently with probability $1/\sqrt{n}$.
The expected size of $X$ is $\sqrt{n}$, and if $x \in X$ is the
closest point
to $F$ in $X$, then the expected $F$-rank of $x$ is $\sqrt{n}$. Set
$\alpha = d(x,F)/(\lowdimapprox)$. We query
$D$ with $F$ and $\alpha$ to obtain a set $R$.
If $d(P,F) \leq \alpha$, then
$R$ contains the nearest neighbor. Otherwise, $x$ is a
$(\lowdimapprox)$-approximate nearest neighbor for $F$.
Thus, it suffices to return the nearest neighbor in $R \cup \{x\}$.
Since with high probability all
points in $R$ have $F$-rank at most $O(\sqrt{n}\log n)$,
we have $|R| = O(\sqrt{n}\log n)$, and the query time is
$O(n^{k/(k+1)}\log^{d-k-1} n)$.
This establishes the following corollary of
Theorem~\ref{thm:lowdimstructure}.
\begin{corollary}
\label{cor:lowdimapprox}
Let $P \subset \R^d$ be an $n$-point set.
We can preprocess $P$ into an
$O(n\log^{d-k-1} n)$ space data structure
for approximate $k$-flat nearest neighbor queries:
given a flat $F$, find a point
$p \in P$ with $d(p,F) \leq (\lowdimapprox) d(P,F)$.
The expected query time is
$O(n^{k/(k+1)}\log^{d-k-1} n)$.
\end{corollary}

\section{Projection Structures}
\label{sec:projectionstructures}

We now describe how to answer queries of type
\textbf{Q1} and \textbf{Q3} efficiently. Our approach is
to project the points into  random subspace of constant dimension
and to solve the problem there using our data structures from
Theorem~\ref{thm:lowdimstructure} and
Corollary~\ref{cor:lowdimapprox}.
For this, we need a Johnson-Lindenstrauss-type lemma that bounds the
distortion, see Section~\ref{sec:dimReduct}.

Let $0 < t \leq 2/(2+40k)$ be a parameter
and let $P \subset R^d$ be a high dimensional $n$-point set.
Set $d' = 2/t + 2$ and
let $M \in \R^{d'\times d}$ be a random projection from
$\R^{d}$ to $\R^{d'}$, scaled by $\sqrt{d/4d'}$.
We obtain $\bar{P} \subset \R^{d'}$ by projecting $P$ using $M$.
We build for $\bar{P}$ the
data structure $D_1$ from Corollary~\ref{cor:lowdimapprox} to answer
\textbf{Q1} queries and $D_2$ from Theorem~\ref{thm:lowdimstructure}
to answer \textbf{Q3} queries.
This needs $O(n \log^{O(d')}n) = O(n\log^{O(1/t)}n)$ space.
For each $p \in P$ we write $\bar{p}$ for the
$d'$-dimensional point $Mp$ and $\bar{F}$ for the projected
flat $MF$.

\subsection{Dimension Reduction}\label{sec:dimReduct}

We use the following variant of the Johnson-Lindenstrauss-Lemma,
as proved by Dasgupta and Gupta~\cite[Lemma~2.2]{DasguptaGu03}.

\begin{lemma}[JL-Lemma]\label{lem:jl}
  Let $d' < d$, and let $M \in \R^{d'\times d}$ be the projection
  matrix onto a random $d'$-dimensional subspace, scaled by a factor
  of $\sqrt{d/d'}$.
  Then, for every vector $x \in \R^d$ of unit length and every
  $\beta > 1$, we have
  \begin{enumerate}
    \item\label{itm:jl-up}
      $\Pr\left[\|Mx\|^2 \geq \beta\right] \leq
      \exp\left(\tfrac{1}{2}d'  (1 - \beta + \ln \beta)\right)$, and
    \item\label{itm:jl-low}
      $\Pr\left[\|Mx\|^2 \leq 1/\beta\right] \leq
      \exp\left(\tfrac{1}{2}d'(1 - 1/\beta - \ln \beta)\right)
      \leq \exp\left(\tfrac{1}{2}d'(1 - \ln \beta)\right)$.\qed{}
  \end{enumerate}
\end{lemma}

\begin{lemma}\label{lem:jl-kflat}
  Let $p \in \R^d$ be a point and let $F \subset \R^d$ be a $k$-flat.
  For $d' \in \{40k, \dots, d-1\}$, let $M \in \R^{d' \times d}$ be
  the projection matrix into a random $d'$-dimensional
  subspace, scaled by $\sqrt{d/4d'}$. Let $\bar{p} = Mp$ and
  $\bar{F} = MF$
  be the projections of $p$ and of $F$, respectively.
  Then, for any $\beta \geq 40k$,
     (i)
         $\Pr[\dist(\bar{p}, \bar{F}) \leq \dist(p, F)] \geq 1 - e^{-d'/2}$; and (ii)
         $\Pr[\dist(\bar{p}, \bar{F}) \geq \dist(p, F) /\beta] \geq 1 - \beta^{-d'/2}$.
\end{lemma}
\begin{proof}
  Let $N = 2M$, and set $q = Np$ and $K = NF$.
  Defining $\Delta_p = \dist(p, F)$ and $\Delta_{q} = \dist(q, K)$,
  we must bound the probabilities
  $\Pr[\Delta_{q} \leq 2\Delta_p]$ for (i) and
  $\Pr[\Delta_{q} \geq 2\Delta_p/\beta]$ for (ii).

  We begin with (i).  Let $p^\parallel$
  be the orthogonal projection of $p$ onto $F$, and let
  $p^\perp = p - p^\parallel$.
  Let $q^\perp = Np^{\perp}$. Then,
  $\Delta_p = \| p^\perp \|$ and $\Delta_q \leq \|q^\perp\|$.
  By Lemma~\ref{lem:jl}(\ref{itm:jl-up}),
  \begin{multline*}
    \Pr\left[\|q^\perp\| \geq 2 \Delta_p\right]
     = \Pr\left[\|Np^\perp\|/\|p^\perp\| \geq 2\right]
     = \Pr\left[\|N(p^\perp/\|p^\perp\|)\|^2 \geq 4\right]\\
     \leq \exp\left(\tfrac{1}{2}d'(1 - 4 + \ln 4)\right)
     \leq \exp(- d'/2).
  \end{multline*}
  Thus, $\Pr[\Delta_q \leq 2\Delta_p] \geq
  \Pr[\|q^\perp \| \leq 2\Delta_p] \geq 1 - \exp(-d'/2)$, as desired.

  For (ii),  choose $k$ orthonormal vectors
  $e_1, \dots, e_k$ such that
  $F = \left\{p^\parallel + \sum_{i=1}^{k} \lambda_i e_i \mid
  \lambda_i \in \R\right\}$.
  Set $u_i = e_i\|p^\perp\|/\beta^2$, and
  consider the lattice
  $L = \left\{ p^\parallel + \sum_{i=1}^k \mu_i u_i \mid
   \mu_i \in \Z \right\} \subset F$.
  Let $\bar{L} = NL$ be the projected lattice.
  We next argue that with high probability
  (i) all points in $\bar{L}$ have
  distance at least $3\|p^\perp\|/\beta$ from $q$; and (ii)
  for $i = 1, \dots, k$, we have $\|Nu_i\| < \|p^\perp\|/\beta\sqrt{k}$.

  To show (i), we partition $L$ into \emph{layers}: for
  $j \in \{0, 1, \dots, \}$, let
  \begin{align*}
    L_j = \left\{p^\parallel + \sum_{i=1}^k \mu_i u_i \mid
      \mu_i \in \{-j, \dots, j\},
     \max_i |\mu_i| = j\right\} \subset L.
  \end{align*}
  Now for any $j \in \N$ and
  $r = p^\parallel + \sum_{i = 1}^k \mu_i u_i \in L_j$,
  Pythagoras gives
  \begin{align*}
      \|p - r\| = \sqrt{\|p^\perp\|^2 +
       \left\|\sum_{i = 1}^k \mu_i u_i \right\|^2}
    = \|p^\perp\|\sqrt{1 + \sum_{i = 1}^k |\mu_i|^2/\beta^4}
    \geq \|p^\perp\|\sqrt{1 + j^2/\beta^4}.
  \end{align*}
  Thus, using Lemma~\ref{lem:jl}(\ref{itm:jl-low}),
  \begin{align*}
    \Pr\left[\|N(p - r)\| \leq 3\|p^\perp\|/\beta\right]
    & = \Pr\left[\|N(p - r)\| / \|p - r\| \leq
        3\|p^\perp\|/\beta\|p - r\|\right] \\
    & \leq \Pr\left[\left\|N (p - r)/\|p - r\|\right\|^2 \leq
      9/\beta^2(1+j^2/\beta^4)\right] \\
    & \leq \exp(\tfrac{1}{2}d'(1 + \ln (9/\beta^2(1+j^2/\beta^4)))) \\
    & \leq (5/\beta)^{d'}(1 + j^2/\beta^4)^{-d'/2},
  \end{align*}
  as $\sqrt{9e} \leq 5$. Now we use a union bound to obtain
  \begin{align*}
    \Pr\left[\exists{r \in L}: \|N(p - r)\| \leq
        3\|p^\perp\|/\beta \right]
    & = \sum_{j = 0}^{\infty}
       \Pr\left[\exists {r \in L_j}: \|N(p - r)\|
           \leq 3\|p^\perp\|/\beta \right]. \\
    \intertext{Grouping the summands into groups of
      $\beta^2$ consecutive terms, this is}
    & = \sum_{l = 0}^{\infty}
        \sum_{j = l\beta^2}^{l\beta^2+\beta^2-1}
       \Pr\left[\exists {r \in L_j}: \|N(p - r)\|
           \leq 3\|p^\perp\|/\beta \right]\\
    & \leq (5/\beta)^{d'} \sum_{l = 0}^{\infty}
      |L_{\leq(l + 1)\beta^2}|(1 + (l\beta^2)^2/\beta^4)^{-d'/2},
    \intertext{where $L_{\leq a} = \bigcup_{i=0}^{a} L_i$.
    Using the rough bound $|L_\leq{a}| \leq (3a)^k$, this is}
    & \leq (5/\beta)^{d'} \sum_{l = 0}^{\infty}
    (3(l + 1)\beta^2)^k (1 + l^2)^{-d'/2}\\
    & = 5^{d'}3^k\beta^{2k - d'} \sum_{l = 0}^{\infty}
    (l + 1)^k (1 + l^2)^{-d'/2}.
    \intertext{For $d' \geq 4k$, we have
       $(l + 1)^k (1 + l^2)^{-d'/2} \leq
      (l + 1)^k (1 + l^2)^{-2k} \leq 1/(1+l^2)$, so we can bound
      the sum by $\sum_{l = 1}^\infty 1/(1 + l^{2}) \leq \pi^2/6$.
    Thus, we have derived}
    \Pr\left[\exists r \in L: \|N(p - r)\| \leq
        3\|p^\perp\|/\beta \right]
      &\leq 5^{d'}3^k\beta^{2k - d'} (\pi^2/6) \leq
        5^{d' + k}\beta^{2k - d'},
     \addtocounter{equation}{1}\tag{\theequation}\label{equ:bound_i}
    \intertext{since $\pi^2/6 \leq 5/3$.}
  \end{align*}
  To show (ii), we use a union bound with
  Lemma~\ref{lem:jl}(\ref{itm:jl-up}). Recalling
  $\|u_1\| = \|p^\perp\|/\beta^2$,
  \begin{align*}
    \Pr\left[\exists i = 1,\dots, k: \|Nu_i\| > \|p^\perp\|/\sqrt{k}\beta
      \right]
     & \leq k\Pr\left[\|Nu_1\| > \|p^\perp\|/\sqrt{k}\beta\right] \\
  &=  k\Pr[\|N(u_1/\|u_1\|)\|^2 > \beta^2/k] \\
  & \leq k \exp\left(\tfrac{1}{2}d'(1 - \beta^2/k + \ln(\beta^2/k))\right) \\
  & \leq k \exp\left(-\beta^2d'/4k\right),
   \addtocounter{equation}{1}\tag{\theequation}\label{equ:bound_ii}
  \end{align*}
  since $c^2/k \geq 2(1 + \ln (\beta^2/k))$ for $\beta^2/k \geq 6$.
  By (\ref{equ:bound_i}) and (\ref{equ:bound_ii}), and recalling
  $\beta, d' \geq 40k$,
  the probability that events (i) and (ii) do not both happen
  is at most
  \begin{multline*}
    5^{d' + k}\beta^{2k - d'} + ke^{-\beta^2d'/4k}  \leq
    5^{(1 + 1/40)d'}\beta^{(1/20 - 1) d'} +
    (\beta/40)e^{-10 cd'} \\
    \leq
    \left(\frac{5^{41/40}}{40^{9/20}}\right)^{d'} \beta^{-d'/2}
    + \frac{1}{10} \beta^{-d'/2}
\leq \beta^{-d'/2}.
\end{multline*}
Suppose (i) and (ii) happen. Fix a point
$w \in K$, and let $\bar{r} \in \bar{L}$ be the point in the projected
lattice that is closest to $w$.
By (i), $\dist(q, \bar{r}) > 3\|p^\perp\|/\beta$.
By (ii)
and the choice of $\bar{r}$,
the $k$-dimensional cube with center $\bar{r}$ and side length
$\|p^\perp\|/\beta\sqrt{k}$ contains $w$. This cube has diameter
$\|p^\perp\|/\beta$.
By triangle inequality,
$\dist(q, w) > \dist(q, \bar{r}) - \dist(\bar{r}, w) \geq
(3/\beta - 1/\beta)\|p^\perp\| = 2\|p^\perp\|/\beta$.
\end{proof}

\subsection{Queries of Type Q1}
Let a query flat $F$ be given.
To answer \textbf{Q1} queries, we compute $\bar{F}$ and query
$D_1$ with $\bar{F}$ to obtain a $(\lowdimapprox['])$-nearest
neighbor $\bar{p}$. We return the original point $p$.
To obtain Theorem~\ref{thm:q1}, we argue that if $\bar{p}$ is a
$(\lowdimapprox)$-nearest neighbor for $\bar{F}$, then $p$ is
a $n^t$-nearest neighbor for $F$ with high probability.

Let $p^* \in P$ be a point with $d(p^*,F) = d(P,F)$.
Set $\delta_{p^*} = d(p^*,F)$ and
$\bar{\delta}_{p^*} = d(\bar{p^*}, \bar{F})$.
Denote by $A_1$ the event that $\bar{\delta}_{p^*} \leq \delta_{p^*}$.
By Lemma~\ref{lem:jl-kflat},
$\Pr[A_1] \geq 1 - e^{-d'/2} = 1 - e^{-1/t - 1}$.
Let $A_2$ be the event that for all points $p \in P$ with
$\delta_{p} = d(p, F) > n^t \delta_{p^*}$ we have
$\bar{\delta}_{p} = d(\bar{p},\bar{F}) >
(\lowdimapprox['])\delta_{p^*}$.
For a fixed $p \in P$, by setting $\beta = n^t/(\lowdimapprox['])$
in Lemma~\ref{lem:jl-kflat},
this probability is

\begin{align*}
  \Pr[\bar{\delta}_{p} > (\lowdimapprox['])\delta_{p^*}]  &
  \geq 1 - (n^t/(\lowdimapprox[']))^{-d'/2} \\
  & = 1 - n^{-1 - t}((4k+3)(2t+1-k)\sqrt{k+1})^{1/t+1} \\
  &\geq   1  - n^{-1 - t/2},
\end{align*}
for $n$ large enough. By the union bound, we
get $\Pr[A_2] \geq 1 - n^{-t/2}$, so the event $A_1 \cap A_2$ occurs
with constant probability. Then, $p$ is a
$n^t$-approximate nearest neighbor for $F$, as desired.

\subsection{Queries of Type Q3}

To answer a query of type \textbf{Q3},
we compute the projection $\bar{F}$ and query $D_2$ with parameter
$\alpha$. We obtain
a set $\bar{R} \subset \bar{P}$ in time
$O(n^{k/(k+1)}\log^{O(1/t)} n + |\bar{R}|)$.
Let $R \subset P$ be the corresponding $d$-dimensional set.
We return a point $p \in R$ that minimizes $d(p,F)$.
If $\delta_{p*} \leq \alpha$, the event $A_1$ from above implies
that $\bar{p^*} \in \bar{R}$, and we correctly return $p^*$.

To bound the size of $|\bar{R}|$, and thus the running time, we use
that $P$ is $\alpha n^t/(2k+1)$-cluster-free.
Let $A_3$ be the event that for all $p \in P$ with
$d(p,F) > \alpha n^t/(2k+1)$, we have
$d(\bar{p}, \bar{F}) > (\lowdimapprox['])\alpha$.
By the definition of cluster-freeness and the guarantee of
Theorem~\ref{thm:lowdimstructure}, we have $|\bar{R}| = m$ in
the case of $A_3$.
Using $\beta = n^t/((2k+1)(\lowdimapprox))$ in Lemma~\ref{lem:jl-kflat}
and doing
a similar calculation as above yields again $\Pr[A_3] \geq 1 - n^{-t/2}$.
Thus, we can answer queries of type \textbf{Q3} successfully
in time $O(n^{k/(k+1)}\log^{O(1/t)}n + m)$ with constant probability,
as claimed in
Theorem~\ref{thm:q3}.

\section{Conclusion}
We have described the first provably efficient
data structure for general $k$-ANN.
Our main technical contribution consists of two
new data structures: the cluster data structure
for high-dimensional $k$-ANN queries, and the
projection data structure for $k$-ANN queries
in constant dimension. We have only
presented the latter structure for a constant approximation factor
$\lowdimapprox$, but we believe that it is possible to extend
it to any fixed approximation factor $c > 1$.
For this, one would need to subdivide the slab structures by
a sufficiently fine sequence of parallel hyperplanes.

Naturally, the most pressing open question is to improve the
query time of our data structure. Also, a further generalization
to more general query or data objects would be of interest.

\section*{Acknowledgments}
This work was initiated while WM, PS, and YS
were visiting the Courant Institute of Mathematical Sciences.
We would like to thank our host Esther Ezra for her hospitality
and many enlightening discussions.

\bibliographystyle{abbrv}
\bibliography{literature}

\end{document}